\documentclass[a4paper,twocolumn,superscriptaddress,11pt,accepted=2019-03-29]{quantumarticle}
\pdfoutput=1
\usepackage[utf8]{inputenc}
\usepackage[english]{babel}
\usepackage[T1]{fontenc}
\usepackage{amsmath}
\usepackage{hyperref}
\usepackage{tikz}
\usepackage[T1]{fontenc}
\usepackage{amsmath, amsthm, amssymb, mathtools, dsfont}
\usepackage{graphicx}
\usepackage{dcolumn}
\usepackage{bm,bbm}
\usepackage{color}
\usepackage{paralist}
\usepackage{tikz}
\usepackage[numbers,sort&compress]{natbib}
\usepackage{changes}
\usepackage{geometry}
\definecolor{myurlcolor}{rgb}{0,0,0.7}
\definecolor{myrefcolor}{rgb}{0.8,0,0}
\DeclareGraphicsExtensions{.png,.pdf,.eps}
\graphicspath{{./figures/} }
\usepackage{epstopdf}

\usepackage{changes}

\makeatletter

\usepackage{tikz}

\usepackage{geometry}
\graphicspath{ {./figures/} }
\usepackage{epstopdf}

\makeatletter

\theoremstyle{plain}
 
 \theoremstyle{plain}
 \newtheorem{lem}{Lemma}
 \theoremstyle{plain}
 \newtheorem{prop}{Proposition}
 \theoremstyle{plain}
 \newtheorem{thm}{Theorem}
 \theoremstyle{plain}
 
 \theoremstyle{plain}
 
 \theoremstyle{plain}
 
  \theoremstyle{plain}
 
 \theoremstyle{remark}
 \newtheorem*{rem*}{Remark}
 \theoremstyle{plain}
  
\newcommand{\ket}[1]{| #1 \rangle}
\newcommand{\bra}[1]{\langle #1 |}

\newcommand{\1}{\mathbbm{1}} 

\renewcommand{\exp}{\mathrm{exp}}
\DeclareMathOperator{\tr}{tr}

\renewcommand{\H}{\mathcal{H}}

\newcommand{\defeq}{\coloneqq}

\newcommand{\C}{\mathbb{C}} 
\newcommand{\R}{\mathbb{R}} 
\newcommand{\M}{\mathbf{M}} 
\newcommand{\N}{\mathbf{N}} 
\newcommand{\W}{\mathbf{W}} 
\newcommand{\X}{\mathbf{X}} 
\newcommand{\FF}{\mathbf{F}} 

\newcommand{\PP}{\mathcal{P}} 
\newcommand{\Sep}{\mathrm{Sep}} 

\newcommand{\E}{\mathcal{E}} 
\newcommand{\F}{\mathcal{F}} 
\newcommand{\Q}{\mathcal{Q}} 
\newcommand{\OO}{\mathcal{O}} 
\newcommand{\IC}{\mathcal{IC}} 

\newcommand{\psucc}{p_{\mathrm{succ}}} 

\newcommand{\Herm}{\mathrm{Herm}} 

\begin{document}

\title{Operational relevance of resource theories of quantum measurements}

\author{Micha\l\ Oszmaniec}
\email{michal.oszmaniec@gmail.com} 
\author{Tanmoy Biswas}
\affiliation{ 
Institute of Theoretical Physics and Astrophysics, National Quantum Information Centre, Faculty of Mathematics, Physics and
Informatics, University of Gdansk, Wita Stwosza 57, 80-308 Gda\'nsk, Poland}

\begin{abstract} 
For any resource theory it is essential to identify tasks for which resource objects offer advantage over free objects. We show that this identification can always be accomplished for resource theories of quantum measurements in which free objects form a convex subset of measurements on a given Hilbert space. To this aim we prove that every resourceful measurement offers advantage for some quantum state discrimination task. Moreover, we give an operational interpretation of robustness, which quantifies the minimal amount of noise that must be added to a measurement to make it free. Specifically, we show that this geometric quantity is related to the maximal relative advantage that a resourceful measurement offers in a class of minimal-error state discrimination (MESD) problems. Finally, we apply our results to two classes of free measurements:  incoherent measurements (measurements that are diagonal in the fixed basis) and separable measurements (measurements whose effects are separable operators).  For both of these scenarios we find, in the asymptotic setting in which the dimension or the number of particles increase to infinity, the maximal relative advantage that resourceful measurements offer for state discrimination tasks. 
\end{abstract}
\maketitle

\section{Introduction}Resource theories \cite{GourResource} constitute a powerful toolbox to study physical systems in the presence of limitations resulting from experimental or operational constrains on the ability to address and manipulate physical systems. This mathematical framework is general enough to encompass both classical and quantum physics, or even more general theories \cite{CoeckeCategory}. In recent years it has been successfully applied to classical and quantum thermodynamics \cite{BrandaoTherm,Goold2016}, processing of quantum information in distributed scenarios \cite{PlenioEnt,ENTrev}, contextuality  \cite{GrudkaContextuality}, nonlocality \cite{BarrettNonLocality}, steering \cite{GallegoSteering2015} and, last but not least, magic state distillation paradigm of quantum computation \cite{VeitchMagic}. On the general level, all resource theories are defined by specifying \emph{free objects} and \emph{free operations}. Free objects form a subset of the set of all objects relevant for the physical situation in question. Likewise, free operations form a subclass of all relevant physical transformations. Typically objects that are not free are called \emph{resource objects}. The specific choice of free objects and free operations depends on the physical context. For example, in entanglement theory the relevant resource theory is based on the principle of locality: set of free objects consists of separable states while free operations are all local operations assisted by classical communication (LOCC).

For any resource theory it is desirable to give operational interpretation of resource objects i.e to identify a task for which a given resource object would prove advantageous over all free objects. The main purpose of this work is to provide such interpretations for resource theories in which a set of free objects consists of convex subset of the set of quantum measurements (POVMs) on a relevant Hilbert space. We realize this goal by showing, that in this setting for every resourceful measurement $\M$ there an instance of minimal-error quantum state discrimination game for which $\M$ gives greater success probability then all free measurements. Quantum state discrimination \cite{Barnett09,StructureStateDISC} is a popular quantum information subroutine that finds applications in different areas of quantum information science including quantum communication \cite{BrunnerQSD}, quantum metrology \cite{SpekkensMetro}, nonlocality \cite{BaeNSQSD} or quantum computation \cite{HSPreview}. We push our operational interpretation further by proving the quantitative relation between the relative advantage of resourceful measurement for state discrimination and the geometric measure of resourcefulness called robustness \cite{Skrzypczyk_robustness,Adesso_subchannel}, which quantifies the minimal amount of noise that has to be added to a POVM to make it free. Previously,  the measure robustness was introduced \cite{Vidal_robustness} to mathematically quantify entanglement. Leter, it was employed in the general resource- theoretic framework  \cite{BrandaoGourRes,Regula_2017,GourResource}.

We apply our general results to two classes of free measurements: incoherent measurement and separable measurements. Incoherent measurements are POVM analogues of incoherent states \cite{AbergSuperposition,PlenioCoherence,ReviewCoherence} and can be understood as measurements originating from a single projective measurement (in the fixed basis) followed by arbitrary classical post-processing. Separable measurements \cite{BandSep,WatrousBook,BaganPurity,KarginHypo}, on the other hand, are POVMs on composite quantum systems whose effects are separable operators. This class contains the set of LOCC measurements \cite{Nathanson2005,Chitambar2014} i.e. measurements that can be implemented via local measurements and LOCC operations. For both incoherent and separable measurements we focus on the asymptotic setting in which the dimension of the system or the number of particles involved go to infinity. In this regime we identify, for both classes of measurements, the maximal relative advantage that resourceful measurements can offer for quantum state discrimination tasks. 

Traditionally, operational interpretations of quantum resources were developed case-by-case for different resource theories \cite{DarianoMetro,HorodeckiFund,SpekkensParity,Cavalcanti2017,WinterDistillation,BiswasCoherence}.
A recent paper \cite{Adesso_subchannel} gave an unified treatment of this problem for all quantum resource theories in which the set of free objects is convex subset of the set of quantum states. 
Specifically, the authors of this work showed that in this context all resource states offer advantage for some sub-channel discrimination problem (this result was previously obtained for resource theories of entanglement \cite{PianiChannel}, coherence \cite{AdessoRobustness}, steering \cite{PianiSteering} and asymmetry \cite{PianiAssymetry}). Our results are complementary to \cite{Adesso_subchannel}, as they give the operational interpretation of \emph{all convex resource theories of quantum measurements}. Moreover, our work greatly generalizes some of the results of \cite{Skrzypczyk_robustness}, where the analogous analysis was presented for the case where free measurements consisted of maximally uninformative measurements i.e. measurements that do not recover any information about the measured quantum states. 

So far, research in quantum resource theories focussed mainly on quantum states while quantum measurements, despite their importance, did not receive much attention. Previously, resource-theoretic perspective was applied in the context of measurement incompatibility \cite{HeinoComp,Leo2017}, measurement simulability via projective measurements \cite{OszmaniecPOVM,MOPOST}, and simulability in the more general scenarios \cite{Leo2017,Kleinmann2016a}. We believe that our results, especially previously unexplored quantitative relation between state discrimination and robustness, provide new quantitative tools to study the restricted classes of POVMs and, more generally, quantum resource theories concentrated around quantum measurements.

\section{Notation and main concepts}

Throughout the paper we will be interested in POVMs on finite dimensional Hilbert space $\H\approx \C^d$. Such generalized measurements (POVMs) can be understood as tuples $\M=\left(M_{1},\ldots,M_{n}\right)$ of non-negative operators on $\C^d$  satisfying $\sum_{i=1}^{n}M_{i}=\1$, where $n$ is the number of outcomes and $\1$ is the identity on $\C^d$. The operators $M_{i}$ are called the effects of POVM $\M$. According to Born's rule, when a POVM $\M$ is measured on a quantum state $\rho$ the probability of obtaining the outcome $i$ is given by $p_i=\tr\left(M_{i}\rho\right)$. We denote the set of POVMs on $\C^d$ with $n$ outcomes by $\PP\left(d,n\right)$. This set has a natural convex structure \cite{DAriano2005}: given two POVMs $\M,\N\in\PP\left(d,n\right)$, their convex combination $p\M+(1-p)\N$ is the POVM with $i$-th effect given by $\left[p\M+\left(1-p\right)\N\right]_i \defeq pM_{i}+(1-p)N_{i}$. The operation of taking convex combinations of measurements can be operationally realized as performing POVMs $\M$ and $\N$ with certain probabilities and then combining the outcomes. 

Our operational interpretation of \emph{resourcefullness} of quantum measurements will be based on the task of minimal-error state discrimination \cite{Barnett09,StructureStateDISC}.  The purpose of this task is to optimally distinguish quantum states states generated by the ensemble $\E = \{q_i,\rho_i\}_{i=1}^{n}$, where $\{q_i\}_{i=1}^{n}$ is a probability distribution and $\{\rho_i\}_{i=1}^{n}$ is a collection of quantum states.  The success probability of identifying the states generated by $\E$ via a measurement $\M \in \PP(d,n)$ is given by  
\begin{equation}\label{eq:prMESD}
\psucc(\E,\M) = \sum_{i=1}^n q_i\tr(M_i\rho_i)\ . 
\end{equation}
One is often interested in choosing the measurement $\M$, such that $\psucc(\E,\M)$ is maximized. However, if only a restricted class of measurements $\F\subset\PP(d,n)$ is allowed, the maximum $\psucc$ may not be achieved.

\section{Convex resource theories of measurements and measurement robustness}
We now give a minimal formulation of a resource theory of measurements, under the assumption of convexity. In our treatment we will focus on a set of quantum measurements $\F$ and a class of free operations \footnote{Of course, the full treatment of resource theory of measurements requires to take into account more complicated aspects such as composability. We plan to address such questions in the future work.}  $\OO$. We first assume (i) that the set of free measurements $\F$ is a convex and closed subset of  $n$-outcome measurements on $\C^d$ i.e $\F\subset \PP(d,n)$, for some suitable Hilbert space $\H \approx \C^d$. Second, we assume that the class of free operations $\OO$ consists of mappings $\varphi:  \PP(d,n) \rightarrow \PP(d,n)$ that (ii) preserve the set of free measurements i.e. for all $\N\in\F$ we have $\varphi(\N)\in\F$ and (iii) are convex-linear i.e.  $\varphi(p\M +(1-p) \M' )= p \varphi(\M) +(1-p) \varphi(\M')$ for all measurements $\M,\M'\in\PP(d,n)$. \\
\begin{figure}
\begin{center}
 \includegraphics[width=9 cm]{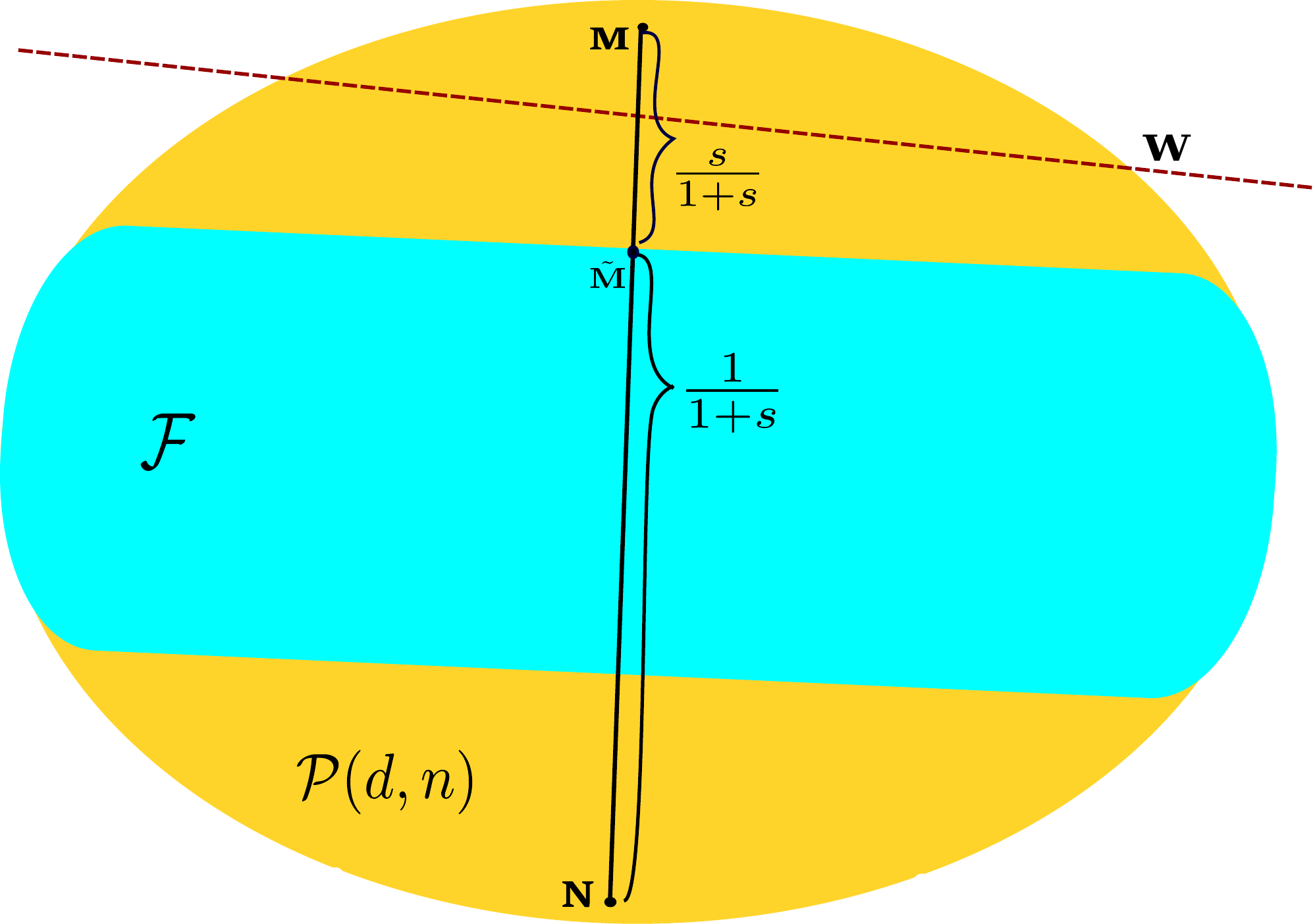}
\caption{\label{fig:Robustness of POVM}
A figure demonstrating some of the key ideas of the paper. The blue part denotes the of all free measurements $\F$ while orange part denotes  the set of all $n$ outcome measurements on $\C^d$ that are not free. (i) Illustration of Theorem \ref{th:witness}: for any   $\M \notin \mathcal{F}$ there always exists a hyperplane $\W$ which separates any POVM $\M \notin \F$ in a sense $\langle \M . \W \rangle > \langle \FF . \W \rangle$ for all $\FF\in\F$. By Theorem \ref{th:witness}, every separating hyperlane can be rewritten in terms of a state discrimination problem  for the ensemble of quantum states $\E$ such that for all free measurements $\FF$ we have $\psucc(\E,\M)> \psucc(\E,\FF)$. (ii) Geometric explanation of Robustness: the straight line represents convex combinations $\frac{\M+s\N}{1+s}$. For a given POVMs $\M$ and $\N$ there exist the minimal value $s_\ast$ (depending on $\N$ for a fixed $\M$) such that $\frac{\M+s\N}{1+s}$ becomes free. The robustness $R_{\F}(\M)$ is defined as the minimum $s_\ast$, considering all possible POVM $\N \in \PP(d,n)$. }  
\end{center}

\end{figure}
\emph{Justification of assumptions:} The convexity of $\F$ can be justified by the fact that, as argued before, convex combination of measurements can be realised by a purely classical process. The compactness of $\F$ is a technical assumption that allows us to formally state our results. However, in basically any physically interesting case the class of free measurements $\F$ can be viewed as a set of convex combinations (convex hull) of a compact set of "primitive"  measurements $\F'$, and hence is naturally compact. Finally, convex linearity of free operations follows from linearity of quantum mechanics while  preservation of the $\F$ by free operation is also natural - operations that are considered free should not be able to create resourceful measurement. In particular, in many cases the set $\F$ contains (a)  classical post-processing i.e. operation of the form $\Q(\M)=\M'$, where $M'_i = \sum_{j} q(i|j) M_i$, for some stochastic matrix $\Q_{ij} = q(i|j)$, and (b) unitary symmetries of the set $\F$. 

For every set of free measurements $\F$ it is possible to introduce a geometric measure quantifying the minimum amount of noise which has to be added to a given POVM $\M$ to make it free.  In analogy to previous works \cite{Vidal_robustness,Adesso_subchannel,Skrzypczyk_robustness} we will call this quantifier \emph{robustness} and denote it by $R_\F$. It is formally defined in the following way (See Fig. \ref{fig:Robustness of POVM}):
\begin{equation} \label{eq:robDEF} 
R_{\F}(\M):=\min \bigg\{s |\ \exists \N \text{ such that } \frac{\M+s\N}{1+s}\in \F \bigg\} \ .
\end{equation}
The robustness enjoys a number of natural properties as stated by the following proposition. 
\begin{prop}[Properties of robustness]\label{lem:prop}
Let $\F\subset \PP(d,n)$ be a closed convex set of free measurements. The robustness  $R_\F$  satisfies the following properties: (i) 
Faithfulness: $R_\F (\M) \geq 0$ and $R_\F (\M)>0$ iff $\M\notin\F$, (ii) Convexity: $R_\F (p\M_1 +(1-p) \M') \leq  p R_\F(\M) +(1-p) R_\F(\M')$ , (iii)  Monotonicity under free operations: $R_\F \left[\varphi(\M)\right]\leq R_\F (\M)$ for all free operations $\varphi\in\OO$.
\end{prop}
The proof of the above follows directly form the definitions of the concepts involved and we defer it to the Appendix.

\section{Main results}
In this part we give our main results concerning the operational interpretation of resourceful measurements for a class of quantum state discrimination tasks. 

\begin{thm}[Minimal-error state discrimination \\ can certify resource character of measurements] \label{th:witness}
For any \textit{resourceful} POVM $\M \notin \F$, there exist an ensemble $\E_\M =\{p_i,\rho_i\}$ such that
\begin{equation} \label{eq:advantageRES}
  \psucc(\E_\M,\M)> \max_{\N\in\F}\psucc(\E_\M,\N)\ .
\end{equation}
\end{thm}
\begin{proof}
We begin the proof by noting that $\PP(d,n)$ is a convex subset of  $(\Herm(\mathcal{H}))^{\times n}$ where $\mathcal{H}\approx\C^d$ is a $d$- dimensional Hilbert space and $\Herm(\mathcal{H})$ is the set of hermitian operator from $\mathcal{H}$ to itself. The space $(\Herm(\mathcal{H}))^{\times n}$ can be regarded as $nd^2$ dimensional real vector space with  the inner product defined by $\langle\textbf{A},\textbf{B}\rangle= \sum_{i=1}^n \tr(A_i B_i)$, where $\textbf{A}, \textbf{B}\in (\Herm(\mathcal{H}))^{\times n} $.  We have assumed that $\F$ is a compact set and hence, due to hyperplane separation theorem (see Fig. \ref{fig:Robustness of POVM}), for any $\M \notin \F$, there exists a tuple of Hermitian operators $\W=\left(W_1,W_2,....W_n\right)$ such that $\sum_{i=1}^n \tr(W_iM_i)>0$ while for all $\N\in\F$ we have  $\sum_i \tr(W_i N_i)\leq 0$. We now define $\tilde{W_i}\coloneqq W_i+|\lambda|\mathbb{I}$, where $\lambda$ is the smallest eigenvalue of operators $\lbrace W_i\rbrace$. We note that for every $i$ the operator $\tilde{W_i}$ is positive-semi-definite and moreover $\sum_i \tr(\tilde{W_i}N_i) < \sum_i \tr(\tilde{W_i}M_i)$. Using this we construct the ensemble  $\E_\M = \bigg\{\frac{\tr(\tilde{W_i})}{\sum_i\tr(\tilde{W_i})},\frac{\tilde{W_i}}{\tr(\tilde{W_i})}\bigg\}$, for which the inequality \eqref{eq:advantageRES} holds. 
\end{proof}

The above result shows the qualitative advantage of resourceful measurements for a class of state-discrimination tasks. In what follows we show a natural way to quantify this advantage. In Theorem \ref{th:opROB} below we show that under the same assumptions as above, the robustness  $\R_\F(\M)$ is related to the maximal relative advantage that a resourceful measurement $\M$ offers for MESD.  To prove this result we will use the following technical Lemma that can be regarded as dual characterisation of measurement robustness.  Its proof relies on duality of semi-infinite convex optimization programs \cite{Brandao_witness} and we present it in the Appendix.

\begin{lem}\label{lem:AltROB}
Let $\M$ $\in$ $\PP(d,n)$ and $\F\subset\PP(d,n)$ be a compact convex set of free measurements. Then, the generalized robustness  $R_{\F}(\M)$ can be expressed as the following optimisation problem.
\begin{equation} \label{eq:AltRob}
\begin{aligned}
& {\text{maximize}}
& & \sum_{i=1}^n \tr(Z_iM_i)-1 \\
& \text{subject to}
& &  Z_i \geq 0, \; i = 1, \ldots, n\ ,\\ 
& & & \sum_{i=1}^n Tr(Z_iN_i) \leq 1  \; \forall  \N  \in \F
\end{aligned}
\end{equation}
\end{lem}
We are now ready to state and prove our main result.
\begin{thm}
[Operational interpretation of robustness in terms of quantum state discrimination] \label{th:opROB} For any class of free measurements $\F$ and any POVM $\M \in \PP(d, n)$ the generalized robustness can be expressed as 
 \begin{equation} \label{eq:opROB}
      R_{\F}(\M) = \max_{\E}\frac{\psucc(\E,\M)}{{\max_{\N\in\F}}\psucc(\E,\N)} -1,
\end{equation}
where the outer optimisation is over $n$ element ensembles of quantum states.
\end{thm} 

\begin{proof}
We begin by showing that LHS\eqref{eq:opROB} $\geq$ RHS\eqref{eq:opROB}. From the definition of generalized robustness from \eqref{eq:robDEF} it follows that for each effect $M_i$ of POVM $\M \in \PP(d,n)$  it holds that $ M_i+\mathcal{R}_{\F}(\M)M'_i=(1+\mathcal{R}_{\F}(\M))N'_i$, where $M'_i$ are the effects of POVM $\M'\in\mathcal{P}(d,n)$ and $N'_i$ are the effects of POVM $\N'\in\F$. Therefore, we have $M_i\leq(1+\mathcal{R}_{\F}(\M))N'_i$ for every $i \in \lbrace 1,n \rbrace$ and therefore
\begin{equation} \label{eq:INEQ1}
\begin{aligned}
\psucc(\E,\M) \leq (1+\mathcal{R}_{\F}(\M))\psucc(\E,\N') \\  \leq (1+\mathcal{R}_{\F}(\M))\max_{\N \in \F}\psucc(\E,\N)\ ,
\end{aligned}
\end{equation}
where the second inequality follows from $\N'\in \F$. 

To show LHS\eqref{eq:opROB} $\leq$ RHS\eqref{eq:opROB}, we use the dual characterization of the generalized robustness from Lemma \ref{lem:AltROB}. Let $\left(Z^\ast_1,...,Z^\ast_n\right)$ be the optimal tuple for which \eqref{eq:AltRob} is maximized. We now define 
\begin{equation}
\E_\ast \coloneqq \bigg\{\frac{\tr(Z^\ast_i)}{\sum_i\tr(Z^\ast_i)},\frac{Z^\ast_i}{\tr(Z^\ast_i)}\bigg\}\ .
\end{equation}
Simple algebraic manipulations give $\mathcal{R}_{\F}(\M)+1= \bigg(\sum_{i=1}^{n}\tr(Z^\ast_i)\bigg) \psucc(\E_\ast,\M)$. On the other hand from \eqref{eq:AltRob} it follows that $\forall \N\in\F$ we have $\psucc(\E_\ast,\N) \leq \frac{1}{\sum_i \tr(Z_i^\ast)}$. Using these two properties we obtain 
\begin{equation}\label{eq:INEQ2}
\frac{\psucc(\E_\ast,\M)}{\max_{\N\in \F}\psucc(\E_\ast,\N)} \geq \mathcal{R}_{\F}(\M)+1 \ .
\end{equation}
By combining \eqref{eq:INEQ1} and \eqref{eq:INEQ2} we complete the proof. 
\end{proof}

Importantly, Theorems \ref{th:witness} and \ref{th:opROB} are valid \emph{for all} convex and compact sets of free measurements $\F$. In particular, we can directly apply them to the measurements that are simulable by projective measurements  \cite{OszmaniecPOVM,MOPOST}  or $k$-outcome measurements   \cite{Leo2017,Kleinmann2016a}. In what follows, we will consider application of our results to two other classes of free measurements: incoherent measurements and separable measurements. We will focus on obtaining the maximal relative advantage for MESD that can be obtained in these scenarios by resourceful measurements.

\subsection{Incoherent measurements}
Quantum coherence or superposition is one of the salient nonclassical aspects of quantum theory. Recently, this feature of quantum theory was investigated thoroughly form the  resource-theoretic perspective \cite{PlenioCoherence,ReviewCoherence}. Moreover, quantum coherence has been identified as resource for a number of quantum information processing tasks like quantum channel discrimination \cite{PianiChannel}, designing quantum algorithms \cite{HilleryAlgo}, quantum metrology \cite{SpekkensMetro} and quantum thermodynamics  \cite{BrandaoTherm}. Finally, there is a plethora of works on quantification of quantum coherence of quantum states \cite{PlenioCoherence,BrussPOVM} and operations \cite{PlenioEvol,BeraEvolution} and on giving operational interpretation to some quantifiers \cite{WinterDistillation,BiswasCoherence}.

Here we propose to introduce a concept of coherence to the realm of POVMs. Specifically, we define the class of \emph{incoherent measurements} $\IC(d,n)$  consisting of $n$-outcome measurements \textbf{S} on $\C^d$,  whose effects are diagonal in a particular orthogonal basis $\lbrace\ket{i}\rbrace$. In what follows, for the sake of simplicity, we assume $n\geq d$.  Incoherent measurements constitute a POVM analogues of the set of \emph{incoherent states}. More formally, the effects of incoherent measurements \textbf{S} can be expressed as $S_i = \sum_{j=1}^{d} q(i|j)|j\rangle\langle j|$, where $i\in[n]$. In is now easy to see that incoherent measurements can always be obtained via classical post-processing of projective measurements in the standard basis.  

One way to quantify the coherence present in a POVM  is the robustness $R_\IC(\M)$ with respect to the set of incoherent measurements. 
Importantly, $R_\IC(\M)$ satisfies a number of natural properties stated in Lemma \ref{lem:prop}. Here, the class of free operations $\OO_\IC$ can be chosen to contain all diagonal unitaries and classical post-processing \footnote{In fact all quantum channels $\Lambda$ such that $\Lambda^\ast$ is an incoherent channel also preserve $\IC(n,d)$.}.  By applying Theorem \ref{th:opROB} directly to this setting we get that robustness $R_\IC(\M)$, can be identified as the maximal relative advantage of  $\M$ for a suitable state discrimination problem. In the following we will take the complemenary approach based on semi-definite programming  \cite{sdpREV}. From the definition of incoherent measurements it follows that  $R_\IC(\M)$ can be casted as the following semi-definite program (SDP) 

 \begin{equation} \label{eq:SDPincohDef}
\begin{aligned}
& {\text{minimize}}
& & s\\
& \text{s.t}
& & \frac{M_a+sN_a}{1+s}=\sum_i q(a|i)|i\rangle\langle i|  \; \forall a \ , \\ 
& & & N_a \geq 0, \; a = 1, \ldots, n. \; \;  \sum_a N_a=\mathbb{I}\ , \\
& & & \sum_a q(a|i)(1+s)=b \; \;  \forall i \ .
\end{aligned}
\end{equation}

Taking the dual problem we get the following auxiliary result whose proof we defer to the Appendix.

\begin{lem}
[Dual characterization of robustness w.r.t. incoherent measurement] Robustness of a POVM $\M \in \PP(d,n)$ w.r.t set of incoherent measurement $\mathcal{IC}(d,n)$, \textit{i.e}, $R_{\IC}(\M)$ can be expressed as the solution of optimization problem given by 
\begin{equation} \label{DualincohLem}
\begin{aligned}
& {\text{maximize}}
& & \sum_{a=1}^{n}\tr(Z_aM_a) -1\\
& \text{subject to}
& & \forall i,a \; \; Z_a \geq 0\;\;\;\; \langle i |Z_a|i\rangle  = \langle i |Z_n|i\rangle \ ,  \\ 
& & & \tr(Z_n)=1\ .
\end{aligned}
\end{equation} 
\end{lem}

From the above we get a physical characterization of ensembles of states that suffice to capture the maximal advantage of coherent measurements over the incoherent ones.

\begin{prop}[Alternative characterisation of robustness for incoherent measurements] \label{lem:dualINCOH}
Let $\M\in\PP(d,n)$. We have the following equality
\begin{equation}\label{eq:opROBincoh}
 R_{\IC}(\M) = \max_{\E_0}\frac{\psucc(\E_0,\M)}{{\max_{\N\in\IC(d,n)}}\psucc(\E_0,\N)} -1,
\end{equation}
where the outer optimisation is over ensembles $\E_0=\lbrace 1/n, \sigma_i \rbrace_{i=1}^n$ consisting of states that cannot be classically distinguished \cite{KorzDist} (that is, for all $j\in[d]$ and for all $i\in[n]$ we have $\bra{j} \sigma_i \ket{j}=1/d$). Due to this property of $\sigma_i$  for all incoherent measurements $\N$ we have $\psucc(\E_0,\N)=1/n$.
\end{prop}

\begin{proof}
We observe that due to the assumed property of states $\sigma_i$ (for all $k$ $\bra{k} \sigma_i \ket{k}=1/d$), for any POVM $\N \in \IC(d,n)$, $\max_{\N\in\IC(d,n)}\psucc(\E_0,\N)=\frac{1}{n}$. Then RHS of \eqref{eq:opROBincoh} is exactly the same as optimization problem given in \eqref{DualincohLem} which completes the proof.
\end{proof}
The above operational characterisation of $R_{\IC}$ allows us to identify measurements with the greatest possible robustness (hence relative advantage for MESD) in a given dimension $d$.

\begin{prop}[Maximal advantage of coherent measurements over incoherent measurements]\label{lem:maxCOH}
The following equality holds
\begin{equation}
\max_{\M\in P(d,n)}  R_{\IC}(\M)=\min\lbrace d,n\rbrace -1\ .
\end{equation}
Moreover, when $n\geq d$ the measurement $\M$ attains the maximal value of $\R_\F$ if and only if its effects are proportional to maximally coherent states i.e pure states of the form $\ket{\psi}= (1/\sqrt{d})\sum_{j=1}^d \exp(i \theta_j) \ket{j}$.
\end{prop}
\begin{proof}[Proof]
From Lemma \ref{lem:dualINCOH} it follows that $R_\IC(\M)+1=\sum_i \tr(M_i \sigma_i)$, where   $\sigma_i$ cannot be classically distinguished. Since for any $\rho$ we have $\tr(M_i \rho)\leq \tr(M_i)$ we obtain $R_\IC(\M)+1\leq \sum_i \tr(M_i)=d$. Moreover, if $n\geq d$  this bound can be easily saturated for projective measurements in the Fourier basis $\ket{\tilde{j}} = \frac{1}{\sqrt{d}}\sum_{k=0}^{d-1} \exp\left(\frac{i2\pi jk}{d}\right)\ket{k}$ =with ensemble $\mathcal{E}_0=\{\frac{1}{d},|\tilde{j}\rangle\}_{j=1}^{d}$.\\
 For $n<d$, we will use \eqref{DualincohLem} to get the missing upper bound. We first note that since operators $M_a$ are upper bounded by $\1$ we have  $\tr(Z_aM_a) \leq \tr (Z_a)=1$. Therefore , from \eqref{DualincohLem} we can conclude that $R_\IC(\M)=\sum_a\tr(Z_aM_a) \leq \sum_a\tr (Z_a)=n$, hence $R_{\IC}(\M) \leq n-1$.     
Combining this with the previous case we conclude $ R_{\IC}(\M)=\min\lbrace d,n\rbrace -1$. \\ Just like in the case $n\geq d$ we will proof the saturation of this bound using Fourier transform of the original basis $ \ket{\tilde{j}}= \frac{1}{\sqrt{d}}\sum_{k=0}^{d-1} e^\frac{i2\pi jk}{d}\ket{k}$. To prove the equality for the case $n < d$, we choose an ensemble $\mathcal{\tilde{E}}$ uniformly distributed over first $n$ orthonormal states from ensemble $\mathcal{E}_0$, \textit{i.e}, $\mathcal{\tilde{E}}=\{\frac{1}{n},|\tilde{j}\rangle\}_{j=1}^{n}$ and an $n$ outcome POVM $\tilde{\M}:=\bigg\{\{|\tilde{j}\rangle\langle \tilde{j}|\}_{i=1}^{n-1}, \mathbb{I}-\sum_{j=1}^{n-1}|\tilde{j}\rangle\langle \tilde{j}|\bigg\}$. For this particular case, again RHS of \eqref{eq:opROBincoh} reduce to $n-1$. Combining these two cases we obtain $\max_{\M\in \PP(d,n)}  R_{\IC}(\M)=\min\lbrace d,n\rbrace -1\ $.  

Now we will show that for $n \geq d $ any POVM for which \eqref{eq:opROBincoh} holds, have the effects directly proportional to $|\psi\rangle\langle \psi|$ where $|\psi\rangle = (1/\sqrt{d})\sum_{j=1}^d \exp(i \theta_j)|j\rangle$. This follows from the simple observation that for $n\geq d$, maximization of \eqref{DualincohLem} $\forall a $, $\tr(Z_aM_a)=\tr(M_a) $ requires that $Z_a=P_{\mathrm{Supp}(M_a)}+B_a$, where $\tr(B_a P_{\mathrm{Supp}(M_a)})=0$. Moreover, to satisfy the constraint of \eqref{DualincohLem}, $P_{\mathrm{Supp}(M_a)}$ has to be one dimensional. This implies only feasible solution will be 
\begin{equation} \label{eq:onlysoln}
\begin{aligned} 
M_a = \alpha_a\Pi_a \; \; \; \; Z_a=\Pi_a,
\end{aligned}
\end{equation}  where $\Pi_a$ is rank one projector with trace 1. Then to be a valid POVM,
\begin{equation} \label{eq:validPOVMApp}
    \sum_a \alpha_a \Pi_a = \1 \Rightarrow \sum_a \alpha_a \langle i|\Pi_a|i\rangle  =1 .
\end{equation}  
Combining \eqref{eq:validPOVMApp} and \eqref{DualincohLem} we obtain $\langle i|\Pi_a|i\rangle=\frac{1}{d}$ for every $a$. This leads to the conclusion that the POVM $\M$ with $n\geq d$ outcomes and satisfying $R_{\IC}(\M)=d-1$  must have effects proportional to $|\psi\rangle\langle \psi|$ where $|\psi\rangle = (1/\sqrt{d})\sum_{j=1}^d \exp(i \theta_j)|j\rangle $.
\end{proof}

\subsection{Separable measurements}
We now turn our attention to the measurements that can be performed on a composite quantum systems. In what follows we will study the maximal advantage that  entangled measurements can offer for MESD over so-called separable measurements on multiparticle quantum systems. Separable measurements contain the set of LOCC measurements and are relatively easy to characterize.  For this reason the maximal  advantage with respect to separable measurements is in general smaller then the advantage with respect to LOCC measurements.

Consider a biparticle quantum system $\H_{AB}=\H_A \otimes \H_B$. We define a POVM $\FF$ as separable if all its effects $\{F_k\}_{k=1}^n$  are separable operators i.e. $ F_i= \sum_k Q^k_i \otimes \tilde{Q}^k_i$  where $\{Q^k_i\}_{i=1}^{n} $ and $\{\tilde{Q}^k_i\}_{i=1}^{n} $  are positive operators on $\H_A$ and $\H_B$ respectively. The above definition can be easily extended to system of multiple qubits  $\H_N =(\C^2)^{\otimes N}$.  We denote the set of biparticle separable measurements by $\Sep(AB)$ and separable multiqubit measurements by $\Sep(N)$. Recall that we made the dependence on  number of outcomes $n$ implicit. In what follows we will assume that $n$ is fixed and equals at least the dimension of the relevant Hilbert space ($d_A d_B$ and $2^N$ respectively). It easily follows that $\Sep(AB)$ and $\Sep(N)$  are convex compact subsets of $\PP(d,n)$ for suitable Hilbert spaces. The following Lemma gives the comprehensive answer to the question about the maximal advantage for MESD that entangled measurements can offer over separable measurements in both bipartite and  multiparty scenario. 
\begin{prop}[Maximal advantage of entangled meaurements over separable measurements] \label{lem:relPOWsep2}
Let $d_A$ and $d_B$ denote the dimensions of local spaces in the biparticle Hilbert space $\H_{AB}$. Then, the following inequalities hold
\begin{equation}\label{eq:UBbipart}
\min\lbrace d_A,d_B\rbrace -1 \leq \max_{\M} R_{\Sep(AB)}(\M)  \leq \min\lbrace d_A,d_B\rbrace\ , 
\end{equation} 
where the  maximization is over all measurements on $\H_{AB}$. \\
Moreover, the maximal robustness of entangled qubit measurement increases exponentially with the size of the system. Specifically, for sufficiently large $N$ we have 
\begin{equation}\label{eq:UBmultipart}
c \frac{2^N}{8N^2}-1 \leq \max_{\M} R_{\Sep(N)}(\M)  \leq 2^{\frac{3}{2}N -1} -1\ , 
\end{equation} 
where the  maximization is over all measurements on $\H_N$ and $c> 0.7$.
\end{prop}
\begin{proof}[Sketch of the proof] 
The upper bounds in \eqref{eq:UBbipart} and \eqref{eq:UBmultipart} follow from application of explicit simulation strategy based on  (global) depolarizing map $\Phi_t$. This map, transforms effects of a  measurement $\M$ as follows: $\Phi_t(M_i) = t M_i +(1-t)\frac{\tr(M_i)}{d} \1$. It is easy to check that the resulting operators still form a POVM on a total space. Because this map preserves traces of the effects we can write $\Phi_t(M_i) = \tr(\M_i)\left(\Phi_t(\rho_i)\right)$, where $\rho_i =M_i/\tr(M_i)$ is a quantum state. To get an upper bound on the robustness of $\M$ it that suffices to find lower bond on $t_\ast$ defined as the number for which \emph{all states} $\rho_i$ constructed from the effects of $M_i$ are separable. To find the upper bound we can directly use the results from entanglement theory. Concretely, the problem of how much "white noise" can be added to a quantum state before it becomes separable quantum states has been studied before. Specifically the maximal value of "random robustness" (minimal $s=1/t -1$ such that $\Phi_t(\rho)$ is separable) equals $\min\lbrace d_A,d_B\rbrace - 1$ in the biparticle scenario \cite{Vidal_robustness} and can be upper bounded by $2^{\frac{3}{2}N -1} -1$ for the multipartice case \cite{Gurvits2004}. To obtain the lower bounds we use characterization of the robustness MESD given in Eq.\eqref{eq:opROB} by showing advantage of entangled measurements over  separable ones for particular uniform ensambles of quantum states. for biparticle case we consider exnamble consisting of  orthogonal maximally states \cite{Nathanson2005}, while for multiparticle case we consider ensemble consisting of $2^N$ iid random quantum states  taken from the Haar measure on $(\C^2)^{\otimes N}$ \cite{Montanaro2007} (see Appendix for details). 
\end{proof}

\begin{rem*}
The above results complement existing results  \cite{Hayden2004,Nathanson2005,Matthews2009} on the relative power of LOCC separable, and global state discrimination. Our focus was on the advantage in discrimination of arbitrary ensembles of quantum states while previously the emphasis was put mostly on ensembles consisting of two states (for example in the context of quantum data hiding). 
\end{rem*}

\section{Open problems}
We finish with giving possible directions of further study. First, it is tempting to ask if non-free character of measurements can be certified in a device-independent manner in prepare-and measure scenarios \cite{ArminSelftest,PiotrPOVM}. Second interesting problem concerns finding the maximal that resourceful POVMs advantage for state discrimination problems for other important classes of measurements such as  LOCC measurements \cite{Nathanson2005,Chitambar2014} or projective-simulable measurements \cite{OszmaniecPOVM,MOPOST}. The potential lack of the increasing  separation (as $d\rightarrow\infty$) between projective and general POVMs will likely  have strong consequences for quantum information processing and Bell nonlocality.  Finally, it is natural to ask if analogous operational interpretations hold also for other classes of quantum objects such as quantum channels or quantum instruments.

\begin{acknowledgments}
We thank Zbigniew Pucha\l a,  Micha\l\ Horodecki, Paul Skrzypczyk and Sibasish Ghosh, for 
interesting and stimulating discussions. We are grateful to Filip Maciejewski for careful reading of the manuscript and his valuable comments. We acknowledge the support
of Homing programme of the Foundation for Polish Science co-financed
by the European Union under the European Regional Development Fund. \\
{Note added----} Upon the completion of this work we became aware of two works \cite{GuhnePOVM,Takagi2019} that independently obtained results overlapping with Theorem \ref{th:opROB} presented here. Moreover, similar methods were used in the context of measurement incompatibility \cite{TeikoComp,CavalComp}. 
\end{acknowledgments}

\onecolumngrid
\appendix

\
\part*{Appendix}
In what follows we give proofs that were omitted in the main part of the manuscript. In Part \ref{app:proofGEN} we give proofs concerning the general properties of the robustness $R_\F$ and its dual formulation. In Part \ref{app:proofPART} we give justifications of results specific to particular scenarios of incoherent and separable measurements.

\bibliography{references.bib}

\begin{thebibliography}{67}%
\makeatletter
\providecommand \@ifxundefined [1]{%
 \@ifx{#1\undefined}
}%
\providecommand \@ifnum [1]{%
 \ifnum #1\expandafter \@firstoftwo
 \else \expandafter \@secondoftwo
 \fi
}%
\providecommand \@ifx [1]{%
 \ifx #1\expandafter \@firstoftwo
 \else \expandafter \@secondoftwo
 \fi
}%
\providecommand \natexlab [1]{#1}%
\providecommand \enquote  [1]{``#1''}%
\providecommand \bibnamefont  [1]{#1}%
\providecommand \bibfnamefont [1]{#1}%
\providecommand \citenamefont [1]{#1}%
\providecommand \href@noop [0]{\@secondoftwo}%
\providecommand \href [0]{\begingroup \@sanitize@url \@href}%
\providecommand \@href[1]{\@@startlink{#1}\@@href}%
\providecommand \@@href[1]{\endgroup#1\@@endlink}%
\providecommand \@sanitize@url [0]{\catcode `\\12\catcode `\$12\catcode
  `\&12\catcode `\#12\catcode `\^12\catcode `\_12\catcode `\%12\relax}%
\providecommand \@@startlink[1]{}%
\providecommand \@@endlink[0]{}%
\providecommand \url  [0]{\begingroup\@sanitize@url \@url }%
\providecommand \@url [1]{\endgroup\@href {#1}{\urlprefix }}%
\providecommand \urlprefix  [0]{URL }%
\providecommand \Eprint [0]{\href }%
\providecommand \doibase [0]{http://dx.doi.org/}%
\providecommand \selectlanguage [0]{\@gobble}%
\providecommand \bibinfo  [0]{\@secondoftwo}%
\providecommand \bibfield  [0]{\@secondoftwo}%
\providecommand \translation [1]{[#1]}%
\providecommand \BibitemOpen [0]{}%
\providecommand \bibitemStop [0]{}%
\providecommand \bibitemNoStop [0]{.\EOS\space}%
\providecommand \EOS [0]{\spacefactor3000\relax}%
\providecommand \BibitemShut  [1]{\csname bibitem#1\endcsname}%
\let\auto@bib@innerbib\@empty
\bibitem [{\citenamefont {Chitambar}\ and\ \citenamefont
  {Gour}(2018)}]{GourResource}%
  \BibitemOpen
  \bibfield  {author} {\bibinfo {author} {\bibfnamefont {E.}~\bibnamefont
  {Chitambar}}\ and\ \bibinfo {author} {\bibfnamefont {G.}~\bibnamefont
  {Gour}},\ }\href {https://arxiv.org/abs/1806.06107} {\bibfield  {journal}
  {\bibinfo  {journal} {arXiv:1806.06107}\ } (\bibinfo {year}
  {2018})}\BibitemShut {NoStop}%
\bibitem [{\citenamefont {Coecke}\ \emph {et~al.}(2016)\citenamefont {Coecke},
  \citenamefont {Fritz},\ and\ \citenamefont {Spekkens}}]{CoeckeCategory}%
  \BibitemOpen
  \bibfield  {author} {\bibinfo {author} {\bibfnamefont {B.}~\bibnamefont
  {Coecke}}, \bibinfo {author} {\bibfnamefont {T.}~\bibnamefont {Fritz}}, \
  and\ \bibinfo {author} {\bibfnamefont {R.~W.}\ \bibnamefont {Spekkens}},\
  }\href {https://doi.org/10.1016/j.ic.2016.02.008} {\bibfield
  {journal} {\bibinfo  {journal} {Information and Computation}\ }\textbf
  {\bibinfo {volume} {250}},\ \bibinfo {pages} {59 } (\bibinfo {year}
  {2016})}\BibitemShut {NoStop}%
\bibitem [{\citenamefont {Brand\~ao}\ \emph {et~al.}(2013)\citenamefont
  {Brand\~ao}, \citenamefont {Horodecki}, \citenamefont {Oppenheim},
  \citenamefont {Renes},\ and\ \citenamefont {Spekkens}}]{BrandaoTherm}%
  \BibitemOpen
  \bibfield  {author} {\bibinfo {author} {\bibfnamefont {F.~G. S.~L.}\
  \bibnamefont {Brand\~ao}}, \bibinfo {author} {\bibfnamefont {M.}~\bibnamefont
  {Horodecki}}, \bibinfo {author} {\bibfnamefont {J.}~\bibnamefont
  {Oppenheim}}, \bibinfo {author} {\bibfnamefont {J.~M.}\ \bibnamefont
  {Renes}}, \ and\ \bibinfo {author} {\bibfnamefont {R.~W.}\ \bibnamefont
  {Spekkens}},\ }\href {\doibase 10.1103/PhysRevLett.111.250404} {\bibfield
  {journal} {\bibinfo  {journal} {Phys. Rev. Lett.}\ }\textbf {\bibinfo
  {volume} {111}},\ \bibinfo {pages} {250404} (\bibinfo {year}
  {2013})}\BibitemShut {NoStop}%
\bibitem [{\citenamefont {Goold}\ \emph {et~al.}(2016)\citenamefont {Goold},
  \citenamefont {Huber}, \citenamefont {Riera}, \citenamefont {del Rio},\ and\
  \citenamefont {Skrzypczyk}}]{Goold2016}%
  \BibitemOpen
  \bibfield  {author} {\bibinfo {author} {\bibfnamefont {J.}~\bibnamefont
  {Goold}}, \bibinfo {author} {\bibfnamefont {M.}~\bibnamefont {Huber}},
  \bibinfo {author} {\bibfnamefont {A.}~\bibnamefont {Riera}}, \bibinfo
  {author} {\bibfnamefont {L.}~\bibnamefont {del Rio}}, \ and\ \bibinfo
  {author} {\bibfnamefont {P.}~\bibnamefont {Skrzypczyk}},\ }\href {\doibase
  10.1088/1751-8113/49/14/143001} {\bibfield  {journal} {\bibinfo  {journal}
  {Journal of Physics A: Mathematical and Theoretical}\ }\textbf {\bibinfo
  {volume} {49}},\ \bibinfo {pages} {143001} (\bibinfo {year}
  {2016})}\BibitemShut {NoStop}%
\bibitem [{\citenamefont {Vedral}\ \emph {et~al.}(1997)\citenamefont {Vedral},
  \citenamefont {Plenio}, \citenamefont {Rippin},\ and\ \citenamefont
  {Knight}}]{PlenioEnt}%
  \BibitemOpen
  \bibfield  {author} {\bibinfo {author} {\bibfnamefont {V.}~\bibnamefont
  {Vedral}}, \bibinfo {author} {\bibfnamefont {M.~B.}\ \bibnamefont {Plenio}},
  \bibinfo {author} {\bibfnamefont {M.~A.}\ \bibnamefont {Rippin}}, \ and\
  \bibinfo {author} {\bibfnamefont {P.~L.}\ \bibnamefont {Knight}},\ }\href
  {\doibase 10.1103/PhysRevLett.78.2275} {\bibfield  {journal} {\bibinfo
  {journal} {Phys. Rev. Lett.}\ }\textbf {\bibinfo {volume} {78}},\ \bibinfo
  {pages} {2275} (\bibinfo {year} {1997})}\BibitemShut {NoStop}%
\bibitem [{\citenamefont {Horodecki}\ \emph {et~al.}(2009)\citenamefont
  {Horodecki}, \citenamefont {Horodecki}, \citenamefont {Horodecki},\ and\
  \citenamefont {Horodecki}}]{ENTrev}%
  \BibitemOpen
  \bibfield  {author} {\bibinfo {author} {\bibfnamefont {R.}~\bibnamefont
  {Horodecki}}, \bibinfo {author} {\bibfnamefont {P.}~\bibnamefont
  {Horodecki}}, \bibinfo {author} {\bibfnamefont {M.}~\bibnamefont
  {Horodecki}}, \ and\ \bibinfo {author} {\bibfnamefont {K.}~\bibnamefont
  {Horodecki}},\ }\href {\doibase 10.1103/RevModPhys.81.865} {\bibfield
  {journal} {\bibinfo  {journal} {Rev. Mod. Phys.}\ }\textbf {\bibinfo {volume}
  {81}},\ \bibinfo {pages} {865} (\bibinfo {year} {2009})}\BibitemShut
  {NoStop}%
\bibitem [{\citenamefont {Grudka}\ \emph {et~al.}(2014)\citenamefont {Grudka},
  \citenamefont {Horodecki}, \citenamefont {Horodecki}, \citenamefont
  {Horodecki}, \citenamefont {Horodecki}, \citenamefont {Joshi}, \citenamefont
  {K\l{}obus},\ and\ \citenamefont {W\'ojcik}}]{GrudkaContextuality}%
  \BibitemOpen
  \bibfield  {author} {\bibinfo {author} {\bibfnamefont {A.}~\bibnamefont
  {Grudka}}, \bibinfo {author} {\bibfnamefont {K.}~\bibnamefont {Horodecki}},
  \bibinfo {author} {\bibfnamefont {M.}~\bibnamefont {Horodecki}}, \bibinfo
  {author} {\bibfnamefont {P.}~\bibnamefont {Horodecki}}, \bibinfo {author}
  {\bibfnamefont {R.}~\bibnamefont {Horodecki}}, \bibinfo {author}
  {\bibfnamefont {P.}~\bibnamefont {Joshi}}, \bibinfo {author} {\bibfnamefont
  {W.}~\bibnamefont {K\l{}obus}}, \ and\ \bibinfo {author} {\bibfnamefont
  {A.}~\bibnamefont {W\'ojcik}},\ }\href {\doibase
  10.1103/PhysRevLett.112.120401} {\bibfield  {journal} {\bibinfo  {journal}
  {Phys. Rev. Lett.}\ }\textbf {\bibinfo {volume} {112}},\ \bibinfo {pages}
  {120401} (\bibinfo {year} {2014})}\BibitemShut {NoStop}%
\bibitem [{\citenamefont {Barrett}(2007)}]{BarrettNonLocality}%
  \BibitemOpen
  \bibfield  {author} {\bibinfo {author} {\bibfnamefont {J.}~\bibnamefont
  {Barrett}},\ }\href {\doibase 10.1103/PhysRevA.75.032304} {\bibfield
  {journal} {\bibinfo  {journal} {Phys. Rev. A}\ }\textbf {\bibinfo {volume}
  {75}},\ \bibinfo {pages} {032304} (\bibinfo {year} {2007})}\BibitemShut
  {NoStop}%
\bibitem [{\citenamefont {Gallego}\ and\ \citenamefont
  {Aolita}(2015)}]{GallegoSteering2015}%
  \BibitemOpen
  \bibfield  {author} {\bibinfo {author} {\bibfnamefont {R.}~\bibnamefont
  {Gallego}}\ and\ \bibinfo {author} {\bibfnamefont {L.}~\bibnamefont
  {Aolita}},\ }\href {\doibase 10.1103/PhysRevX.5.041008} {\bibfield  {journal}
  {\bibinfo  {journal} {Phys. Rev. X}\ }\textbf {\bibinfo {volume} {5}},\
  \bibinfo {pages} {041008} (\bibinfo {year} {2015})}\BibitemShut {NoStop}%
\bibitem [{\citenamefont {Veitch}\ \emph {et~al.}(2012)\citenamefont {Veitch},
  \citenamefont {Ferrie}, \citenamefont {Gross},\ and\ \citenamefont
  {Emerson}}]{VeitchMagic}%
  \BibitemOpen
  \bibfield  {author} {\bibinfo {author} {\bibfnamefont {V.}~\bibnamefont
  {Veitch}}, \bibinfo {author} {\bibfnamefont {C.}~\bibnamefont {Ferrie}},
  \bibinfo {author} {\bibfnamefont {D.}~\bibnamefont {Gross}}, \ and\ \bibinfo
  {author} {\bibfnamefont {J.}~\bibnamefont {Emerson}},\ }\href
  {http://stacks.iop.org/1367-2630/14/i=11/a=113011} {\bibfield  {journal}
  {\bibinfo  {journal} {New Journal of Physics}\ }\textbf {\bibinfo {volume}
  {14}},\ \bibinfo {pages} {113011} (\bibinfo {year} {2012})}\BibitemShut
  {NoStop}%
\bibitem [{\citenamefont {Barnett}\ and\ \citenamefont
  {Croke}(2009)}]{Barnett09}%
  \BibitemOpen
  \bibfield  {author} {\bibinfo {author} {\bibfnamefont {S.~M.}\ \bibnamefont
  {Barnett}}\ and\ \bibinfo {author} {\bibfnamefont {S.}~\bibnamefont
  {Croke}},\ }\href
  {https://www.osapublishing.org/aop/abstract.cfm?uri=aop-1-2-238} {\bibfield
  {journal} {\bibinfo  {journal} {Adv. Opt. Photon.}\ }\textbf {\bibinfo
  {volume} {1}},\ \bibinfo {pages} {238} (\bibinfo {year} {2009})}\BibitemShut
  {NoStop}%
\bibitem [{\citenamefont {Bae}\ and\ \citenamefont
  {Kwek}(2015)}]{StructureStateDISC}%
  \BibitemOpen
  \bibfield  {author} {\bibinfo {author} {\bibfnamefont {J.}~\bibnamefont
  {Bae}}\ and\ \bibinfo {author} {\bibfnamefont {L.-C.}\ \bibnamefont {Kwek}},\
  }\href {http://stacks.iop.org/1751-8121/48/i=8/a=083001} {\bibfield
  {journal} {\bibinfo  {journal} {J. Phys. A: Math. Gen.}\ }\textbf {\bibinfo
  {volume} {48}},\ \bibinfo {pages} {083001} (\bibinfo {year}
  {2015})}\BibitemShut {NoStop}%
\bibitem [{\citenamefont {Brunner}\ \emph {et~al.}(2013)\citenamefont
  {Brunner}, \citenamefont {Navascu\'es},\ and\ \citenamefont
  {V\'ertesi}}]{BrunnerQSD}%
  \BibitemOpen
  \bibfield  {author} {\bibinfo {author} {\bibfnamefont {N.}~\bibnamefont
  {Brunner}}, \bibinfo {author} {\bibfnamefont {M.}~\bibnamefont
  {Navascu\'es}}, \ and\ \bibinfo {author} {\bibfnamefont {T.}~\bibnamefont
  {V\'ertesi}},\ }\href {\doibase 10.1103/PhysRevLett.110.150501} {\bibfield
  {journal} {\bibinfo  {journal} {Phys. Rev. Lett.}\ }\textbf {\bibinfo
  {volume} {110}},\ \bibinfo {pages} {150501} (\bibinfo {year}
  {2013})}\BibitemShut {NoStop}%
\bibitem [{\citenamefont {Marvian}\ and\ \citenamefont
  {Spekkens}(2016)}]{SpekkensMetro}%
  \BibitemOpen
  \bibfield  {author} {\bibinfo {author} {\bibfnamefont {I.}~\bibnamefont
  {Marvian}}\ and\ \bibinfo {author} {\bibfnamefont {R.~W.}\ \bibnamefont
  {Spekkens}},\ }\href {\doibase 10.1103/PhysRevA.94.052324} {\bibfield
  {journal} {\bibinfo  {journal} {Phys. Rev. A}\ }\textbf {\bibinfo {volume}
  {94}},\ \bibinfo {pages} {052324} (\bibinfo {year} {2016})}\BibitemShut
  {NoStop}%
\bibitem [{\citenamefont {Bae}\ \emph {et~al.}(2011)\citenamefont {Bae},
  \citenamefont {Hwang},\ and\ \citenamefont {Han}}]{BaeNSQSD}%
  \BibitemOpen
  \bibfield  {author} {\bibinfo {author} {\bibfnamefont {J.}~\bibnamefont
  {Bae}}, \bibinfo {author} {\bibfnamefont {W.-Y.}\ \bibnamefont {Hwang}}, \
  and\ \bibinfo {author} {\bibfnamefont {Y.-D.}\ \bibnamefont {Han}},\ }\href
  {\doibase 10.1103/PhysRevLett.107.170403} {\bibfield  {journal} {\bibinfo
  {journal} {Phys. Rev. Lett.}\ }\textbf {\bibinfo {volume} {107}},\ \bibinfo
  {pages} {170403} (\bibinfo {year} {2011})}\BibitemShut {NoStop}%
\bibitem [{\citenamefont {Lomont}(2005)}]{HSPreview}%
  \BibitemOpen
  \bibfield  {author} {\bibinfo {author} {\bibfnamefont {C.}~\bibnamefont
  {Lomont}},\ }\href {https://arxiv.org/abs/quant-ph/0411037} {\bibfield
  {journal} {\bibinfo  {journal} {arxiv:0411037}\ } (\bibinfo {year}
  {2005})}\BibitemShut {NoStop}%
  \bibitem [{\citenamefont {Skrzypczyk}\ and\ \citenamefont
  {Linden}(2019)}]{Skrzypczyk_robustness}%
  \BibitemOpen
  \bibfield  {author} {\bibinfo {author} {\bibfnamefont {P.}~\bibnamefont
  {Skrzypczyk}}\ and\ \bibinfo {author} {\bibfnamefont {N.}~\bibnamefont
  {Linden}},\ }\href {\doibase 10.1103/PhysRevLett.122.140403}
  {\bibfield  {journal} {\bibinfo  {journal} {Phys. Rev. Lett.}\ }\textbf
  {\bibinfo {volume} {122}},\ \bibinfo {pages} {140403} (\bibinfo {year}
  {2019})}\BibitemShut {NoStop}%
\bibitem [{\citenamefont {Takagi}\ \emph {et~al.}(2018)\citenamefont {Takagi},
  \citenamefont {Regula}, \citenamefont {Bu}, \citenamefont {Liu},\ and\
  \citenamefont {Adesso}}]{Adesso_subchannel}%
  \BibitemOpen
  \bibfield  {author} {\bibinfo {author} {\bibfnamefont {R.}~\bibnamefont
  {Takagi}}, \bibinfo {author} {\bibfnamefont {B.}~\bibnamefont {Regula}},
  \bibinfo {author} {\bibfnamefont {K.}~\bibnamefont {Bu}}, \bibinfo {author}
  {\bibfnamefont {Z.-W.}\ \bibnamefont {Liu}}, \ and\ \bibinfo {author}
  {\bibfnamefont {G.}~\bibnamefont {Adesso}},\ }\href {\doibase 10.1103/PhysRevLett.122.140402}
  {\bibfield  {journal} {\bibinfo  {journal} {Phys. Rev. Lett.}\ }\textbf
  {\bibinfo {volume} {122}},\ \bibinfo {pages} {140402} (\bibinfo {year}
  {2019})}\BibitemShut {NoStop}%
\bibitem [{\citenamefont {Vidal}\ and\ \citenamefont
  {Tarrach}(1999)}]{Vidal_robustness}%
  \BibitemOpen
  \bibfield  {author} {\bibinfo {author} {\bibfnamefont {G.}~\bibnamefont
  {Vidal}}\ and\ \bibinfo {author} {\bibfnamefont {R.}~\bibnamefont
  {Tarrach}},\ }\href {\doibase 10.1103/PhysRevA.72.022310} {\bibfield
  {journal} {\bibinfo  {journal} {Phys. Rev. A}\ }\textbf {\bibinfo {volume}
  {59}},\ \bibinfo {pages} {141} (\bibinfo {year} {1999})}\BibitemShut
  {NoStop}%
\bibitem [{\citenamefont {Brand\~ao}\ and\ \citenamefont
  {Gour}(2015)}]{BrandaoGourRes}%
  \BibitemOpen
  \bibfield  {author} {\bibinfo {author} {\bibfnamefont {F.~G. S.~L.}\
  \bibnamefont {Brand\~ao}}\ and\ \bibinfo {author} {\bibfnamefont
  {G.}~\bibnamefont {Gour}},\ }\href {\doibase 10.1103/PhysRevLett.115.070503}
  {\bibfield  {journal} {\bibinfo  {journal} {Phys. Rev. Lett.}\ }\textbf
  {\bibinfo {volume} {115}},\ \bibinfo {pages} {070503} (\bibinfo {year}
  {2015})}\BibitemShut {NoStop}%
\bibitem [{\citenamefont {Regula}(2017)}]{Regula_2017}%
  \BibitemOpen
  \bibfield  {author} {\bibinfo {author} {\bibfnamefont {B.}~\bibnamefont
  {Regula}},\ }\href {\doibase 10.1088/1751-8121/aa9100} {\bibfield  {journal}
  {\bibinfo  {journal} {Journal of Physics A: Mathematical and Theoretical}\
  }\textbf {\bibinfo {volume} {51}},\ \bibinfo {pages} {045303} (\bibinfo
  {year} {2017})}\BibitemShut {NoStop}%
\bibitem [{\citenamefont {Aberg}(2006)}]{AbergSuperposition}%
  \BibitemOpen
  \bibfield  {author} {\bibinfo {author} {\bibfnamefont {J.}~\bibnamefont
  {Aberg}},\ }\href {https://arxiv.org/abs/quant-ph/0612146} {\bibfield
  {journal} {\bibinfo  {journal} {arXiv:quant-ph/0612146}\ } (\bibinfo {year}
  {2006})}\BibitemShut {NoStop}%
\bibitem [{\citenamefont {Baumgratz}\ \emph {et~al.}(2014)\citenamefont
  {Baumgratz}, \citenamefont {Cramer},\ and\ \citenamefont
  {Plenio}}]{PlenioCoherence}%
  \BibitemOpen
  \bibfield  {author} {\bibinfo {author} {\bibfnamefont {T.}~\bibnamefont
  {Baumgratz}}, \bibinfo {author} {\bibfnamefont {M.}~\bibnamefont {Cramer}}, \
  and\ \bibinfo {author} {\bibfnamefont {M.~B.}\ \bibnamefont {Plenio}},\
  }\href {\doibase 10.1103/PhysRevLett.113.140401} {\bibfield  {journal}
  {\bibinfo  {journal} {Phys. Rev. Lett.}\ }\textbf {\bibinfo {volume} {113}},\
  \bibinfo {pages} {140401} (\bibinfo {year} {2014})}\BibitemShut {NoStop}%
\bibitem [{\citenamefont {Streltsov}\ \emph {et~al.}(2017)\citenamefont
  {Streltsov}, \citenamefont {Adesso},\ and\ \citenamefont
  {Plenio}}]{ReviewCoherence}%
  \BibitemOpen
  \bibfield  {author} {\bibinfo {author} {\bibfnamefont {A.}~\bibnamefont
  {Streltsov}}, \bibinfo {author} {\bibfnamefont {G.}~\bibnamefont {Adesso}}, \
  and\ \bibinfo {author} {\bibfnamefont {M.~B.}\ \bibnamefont {Plenio}},\
  }\href {\doibase 10.1103/RevModPhys.89.041003} {\bibfield  {journal}
  {\bibinfo  {journal} {Rev. Mod. Phys.}\ }\textbf {\bibinfo {volume} {89}},\
  \bibinfo {pages} {041003} (\bibinfo {year} {2017})}\BibitemShut {NoStop}%
\bibitem [{\citenamefont {Bandyopadhyay}\ and\ \citenamefont
  {Nathanson}(2013)}]{BandSep}%
  \BibitemOpen
  \bibfield  {author} {\bibinfo {author} {\bibfnamefont {S.}~\bibnamefont
  {Bandyopadhyay}}\ and\ \bibinfo {author} {\bibfnamefont {M.}~\bibnamefont
  {Nathanson}},\ }\href {\doibase 10.1103/PhysRevA.88.052313} {\bibfield
  {journal} {\bibinfo  {journal} {Phys. Rev. A}\ }\textbf {\bibinfo {volume}
  {88}},\ \bibinfo {pages} {052313} (\bibinfo {year} {2013})}\BibitemShut
  {NoStop}%
\bibitem [{\citenamefont {Watrous}(2018)}]{WatrousBook}%
  \BibitemOpen
  \bibfield  {author} {\bibinfo {author} {\bibfnamefont {J.}~\bibnamefont
  {Watrous}},\ }\href {\doibase 10.1017/9781316848142} {\emph {\bibinfo {title}
  {The Theory of Quantum Information}}}\ (\bibinfo  {publisher} {Cambridge
  University Press},\ \bibinfo {year} {2018})\BibitemShut {NoStop}%
\bibitem [{\citenamefont {Bagan}\ \emph {et~al.}(2005)\citenamefont {Bagan},
  \citenamefont {Ballester}, \citenamefont {Mu\~noz Tapia},\ and\ \citenamefont
  {Romero-Isart}}]{BaganPurity}%
  \BibitemOpen
  \bibfield  {author} {\bibinfo {author} {\bibfnamefont {E.}~\bibnamefont
  {Bagan}}, \bibinfo {author} {\bibfnamefont {M.~A.}\ \bibnamefont
  {Ballester}}, \bibinfo {author} {\bibfnamefont {R.}~\bibnamefont {Mu\~noz
  Tapia}}, \ and\ \bibinfo {author} {\bibfnamefont {O.}~\bibnamefont
  {Romero-Isart}},\ }\href {\doibase 10.1103/PhysRevLett.95.110504} {\bibfield
  {journal} {\bibinfo  {journal} {Phys. Rev. Lett.}\ }\textbf {\bibinfo
  {volume} {95}},\ \bibinfo {pages} {110504} (\bibinfo {year}
  {2005})}\BibitemShut {NoStop}%
\bibitem [{\citenamefont {Kargin}(2004)}]{KarginHypo}%
  \BibitemOpen
  \bibfield  {author} {\bibinfo {author} {\bibfnamefont {V.}~\bibnamefont
  {Kargin}},\ }\href {\doibase 10.1214/009053604000001219} {\bibfield
  {journal} {\bibinfo  {journal} {The Annals of Statistics}\ }\textbf {\bibinfo
  {volume} {33}},\ \bibinfo {pages} {959} (\bibinfo {year} {2004})}\BibitemShut
  {NoStop}%
\bibitem [{\citenamefont {Nathanson}(2005)}]{Nathanson2005}%
  \BibitemOpen
  \bibfield  {author} {\bibinfo {author} {\bibfnamefont {M.}~\bibnamefont
  {Nathanson}},\ }\href {\doibase 10.1063/1.1914731} {\bibfield  {journal}
  {\bibinfo  {journal} {J. Math. Phys.}\ }\textbf {\bibinfo {volume} {46}},\
  \bibinfo {pages} {141} (\bibinfo {year} {2005})}\BibitemShut {NoStop}%
\bibitem [{\citenamefont {Chitambar}\ \emph {et~al.}(2014)\citenamefont
  {Chitambar}, \citenamefont {Leung}, \citenamefont {Man{\v{c}}inska},
  \citenamefont {Ozols},\ and\ \citenamefont {Winter}}]{Chitambar2014}%
  \BibitemOpen
  \bibfield  {author} {\bibinfo {author} {\bibfnamefont {E.}~\bibnamefont
  {Chitambar}}, \bibinfo {author} {\bibfnamefont {D.}~\bibnamefont {Leung}},
  \bibinfo {author} {\bibfnamefont {L.}~\bibnamefont {Man{\v{c}}inska}},
  \bibinfo {author} {\bibfnamefont {M.}~\bibnamefont {Ozols}}, \ and\ \bibinfo
  {author} {\bibfnamefont {A.}~\bibnamefont {Winter}},\ }\href {\doibase
  10.1007/s00220-014-1953-9} {\bibfield  {journal} {\bibinfo  {journal}
  {Communications in Mathematical Physics}\ }\textbf {\bibinfo {volume}
  {328}},\ \bibinfo {pages} {303} (\bibinfo {year} {2014})}\BibitemShut
  {NoStop}%
\bibitem [{\citenamefont {D'Ariano}\ \emph {et~al.}(2001)\citenamefont
  {D'Ariano}, \citenamefont {Lo~Presti},\ and\ \citenamefont
  {Paris}}]{DarianoMetro}%
  \BibitemOpen
  \bibfield  {author} {\bibinfo {author} {\bibfnamefont {G.~M.}\ \bibnamefont
  {D'Ariano}}, \bibinfo {author} {\bibfnamefont {P.}~\bibnamefont {Lo~Presti}},
  \ and\ \bibinfo {author} {\bibfnamefont {M.~G.~A.}\ \bibnamefont {Paris}},\
  }\href {\doibase 10.1103/PhysRevLett.87.270404} {\bibfield  {journal}
  {\bibinfo  {journal} {Phys. Rev. Lett.}\ }\textbf {\bibinfo {volume} {87}},\
  \bibinfo {pages} {270404} (\bibinfo {year} {2001})}\BibitemShut {NoStop}%
\bibitem [{\citenamefont {Horodecki}\ and\ \citenamefont
  {Oppenheim}(2013)}]{HorodeckiFund}%
  \BibitemOpen
  \bibfield  {author} {\bibinfo {author} {\bibfnamefont {M.}~\bibnamefont
  {Horodecki}}\ and\ \bibinfo {author} {\bibfnamefont {J.}~\bibnamefont
  {Oppenheim}},\ }\href {\doibase 10.1038/ncomms3059} {\bibfield  {journal}
  {\bibinfo  {journal} {Nat. Commun.}\ }\textbf {\bibinfo {volume} {4}}
  (\bibinfo {year} {2013}),\ 10.1038/ncomms3059}\BibitemShut {NoStop}%
\bibitem [{\citenamefont {Spekkens}\ \emph {et~al.}(2009)\citenamefont
  {Spekkens}, \citenamefont {Buzacott}, \citenamefont {Keehn}, \citenamefont
  {Toner},\ and\ \citenamefont {Pryde}}]{SpekkensParity}%
  \BibitemOpen
  \bibfield  {author} {\bibinfo {author} {\bibfnamefont {R.~W.}\ \bibnamefont
  {Spekkens}}, \bibinfo {author} {\bibfnamefont {D.~H.}\ \bibnamefont
  {Buzacott}}, \bibinfo {author} {\bibfnamefont {A.~J.}\ \bibnamefont {Keehn}},
  \bibinfo {author} {\bibfnamefont {B.}~\bibnamefont {Toner}}, \ and\ \bibinfo
  {author} {\bibfnamefont {G.~J.}\ \bibnamefont {Pryde}},\ }\href {\doibase
  10.1103/PhysRevLett.102.010401} {\bibfield  {journal} {\bibinfo  {journal}
  {Phys. Rev. Lett.}\ }\textbf {\bibinfo {volume} {102}},\ \bibinfo {pages}
  {010401} (\bibinfo {year} {2009})}\BibitemShut {NoStop}%
\bibitem [{\citenamefont {Cavalcanti}\ \emph {et~al.}(2017)\citenamefont
  {Cavalcanti}, \citenamefont {Skrzypczyk},\ and\ \citenamefont {\ifmmode
  \check{S}\else \v{S}\fi{}upi\ifmmode~\acute{c}\else
  \'{c}\fi{}}}]{Cavalcanti2017}%
  \BibitemOpen
  \bibfield  {author} {\bibinfo {author} {\bibfnamefont {D.}~\bibnamefont
  {Cavalcanti}}, \bibinfo {author} {\bibfnamefont {P.}~\bibnamefont
  {Skrzypczyk}}, \ and\ \bibinfo {author} {\bibfnamefont {I.}~\bibnamefont
  {\ifmmode \check{S}\else \v{S}\fi{}upi\ifmmode~\acute{c}\else \'{c}\fi{}}},\
  }\href {\doibase 10.1103/PhysRevLett.119.110501} {\bibfield  {journal}
  {\bibinfo  {journal} {Phys. Rev. Lett.}\ }\textbf {\bibinfo {volume} {119}},\
  \bibinfo {pages} {110501} (\bibinfo {year} {2017})}\BibitemShut {NoStop}%
\bibitem [{\citenamefont {Winter}\ and\ \citenamefont
  {Yang}(2016)}]{WinterDistillation}%
  \BibitemOpen
  \bibfield  {author} {\bibinfo {author} {\bibfnamefont {A.}~\bibnamefont
  {Winter}}\ and\ \bibinfo {author} {\bibfnamefont {D.}~\bibnamefont {Yang}},\
  }\href {\doibase 10.1103/PhysRevLett.116.120404} {\bibfield  {journal}
  {\bibinfo  {journal} {Phys. Rev. Lett.}\ }\textbf {\bibinfo {volume} {116}},\
  \bibinfo {pages} {120404} (\bibinfo {year} {2016})}\BibitemShut {NoStop}%
\bibitem [{\citenamefont {Biswas}\ \emph {et~al.}(2017)\citenamefont {Biswas},
  \citenamefont {Diaz},\ and\ \citenamefont {Winter}}]{BiswasCoherence}%
  \BibitemOpen
  \bibfield  {author} {\bibinfo {author} {\bibfnamefont {T.}~\bibnamefont
  {Biswas}}, \bibinfo {author} {\bibfnamefont {M.~G.}\ \bibnamefont {Diaz}}, \
  and\ \bibinfo {author} {\bibfnamefont {A.}~\bibnamefont {Winter}},\ }\href
  {\doibase 10.1098/rspa.2017.0170} {\bibfield  {journal} {\bibinfo  {journal}
  {Proc. R. Soc. A}\ }\textbf {\bibinfo {volume} {473}} (\bibinfo {year}
  {2017}),\ 10.1098/rspa.2017.0170}\BibitemShut {NoStop}%
\bibitem [{\citenamefont {Piani}\ \emph
  {et~al.}(2016{\natexlab{a}})\citenamefont {Piani}, \citenamefont
  {Cianciaruso}, \citenamefont {Bromley}, \citenamefont {Napoli}, \citenamefont
  {Johnston},\ and\ \citenamefont {Adesso}}]{PianiChannel}%
  \BibitemOpen
  \bibfield  {author} {\bibinfo {author} {\bibfnamefont {M.}~\bibnamefont
  {Piani}}, \bibinfo {author} {\bibfnamefont {M.}~\bibnamefont {Cianciaruso}},
  \bibinfo {author} {\bibfnamefont {T.~R.}\ \bibnamefont {Bromley}}, \bibinfo
  {author} {\bibfnamefont {C.}~\bibnamefont {Napoli}}, \bibinfo {author}
  {\bibfnamefont {N.}~\bibnamefont {Johnston}}, \ and\ \bibinfo {author}
  {\bibfnamefont {G.}~\bibnamefont {Adesso}},\ }\href {\doibase
  10.1103/PhysRevA.93.042107} {\bibfield  {journal} {\bibinfo  {journal} {Phys.
  Rev. A}\ }\textbf {\bibinfo {volume} {93}},\ \bibinfo {pages} {042107}
  (\bibinfo {year} {2016}{\natexlab{a}})}\BibitemShut {NoStop}%
\bibitem [{\citenamefont {Napoli}\ \emph {et~al.}(2016)\citenamefont {Napoli},
  \citenamefont {Bromley}, \citenamefont {Cianciaruso}, \citenamefont {Piani},
  \citenamefont {Johnston},\ and\ \citenamefont {Adesso}}]{AdessoRobustness}%
  \BibitemOpen
  \bibfield  {author} {\bibinfo {author} {\bibfnamefont {C.}~\bibnamefont
  {Napoli}}, \bibinfo {author} {\bibfnamefont {T.~R.}\ \bibnamefont {Bromley}},
  \bibinfo {author} {\bibfnamefont {M.}~\bibnamefont {Cianciaruso}}, \bibinfo
  {author} {\bibfnamefont {M.}~\bibnamefont {Piani}}, \bibinfo {author}
  {\bibfnamefont {N.}~\bibnamefont {Johnston}}, \ and\ \bibinfo {author}
  {\bibfnamefont {G.}~\bibnamefont {Adesso}},\ }\href {\doibase
  10.1103/PhysRevLett.116.150502} {\bibfield  {journal} {\bibinfo  {journal}
  {Phys. Rev. Lett.}\ }\textbf {\bibinfo {volume} {116}},\ \bibinfo {pages}
  {150502} (\bibinfo {year} {2016})}\BibitemShut {NoStop}%
\bibitem [{\citenamefont {Piani}\ and\ \citenamefont
  {Watrous}(2015)}]{PianiSteering}%
  \BibitemOpen
  \bibfield  {author} {\bibinfo {author} {\bibfnamefont {M.}~\bibnamefont
  {Piani}}\ and\ \bibinfo {author} {\bibfnamefont {J.}~\bibnamefont
  {Watrous}},\ }\href {\doibase 10.1103/PhysRevLett.114.060404} {\bibfield
  {journal} {\bibinfo  {journal} {Phys. Rev. Lett.}\ }\textbf {\bibinfo
  {volume} {114}},\ \bibinfo {pages} {060404} (\bibinfo {year}
  {2015})}\BibitemShut {NoStop}%
\bibitem [{\citenamefont {Piani}\ \emph
  {et~al.}(2016{\natexlab{b}})\citenamefont {Piani}, \citenamefont
  {Cianciaruso}, \citenamefont {Bromley}, \citenamefont {Napoli}, \citenamefont
  {Johnston},\ and\ \citenamefont {Adesso}}]{PianiAssymetry}%
  \BibitemOpen
  \bibfield  {author} {\bibinfo {author} {\bibfnamefont {M.}~\bibnamefont
  {Piani}}, \bibinfo {author} {\bibfnamefont {M.}~\bibnamefont {Cianciaruso}},
  \bibinfo {author} {\bibfnamefont {T.~R.}\ \bibnamefont {Bromley}}, \bibinfo
  {author} {\bibfnamefont {C.}~\bibnamefont {Napoli}}, \bibinfo {author}
  {\bibfnamefont {N.}~\bibnamefont {Johnston}}, \ and\ \bibinfo {author}
  {\bibfnamefont {G.}~\bibnamefont {Adesso}},\ }\href {\doibase
  10.1103/PhysRevA.93.042107} {\bibfield  {journal} {\bibinfo  {journal} {Phys.
  Rev. A}\ }\textbf {\bibinfo {volume} {93}},\ \bibinfo {pages} {042107}
  (\bibinfo {year} {2016}{\natexlab{b}})}\BibitemShut {NoStop}%
\bibitem [{\citenamefont {Heinosaari}\ \emph {et~al.}(2015)\citenamefont
  {Heinosaari}, \citenamefont {Kiukas},\ and\ \citenamefont
  {Reitzner}}]{HeinoComp}%
  \BibitemOpen
  \bibfield  {author} {\bibinfo {author} {\bibfnamefont {T.}~\bibnamefont
  {Heinosaari}}, \bibinfo {author} {\bibfnamefont {J.}~\bibnamefont {Kiukas}},
  \ and\ \bibinfo {author} {\bibfnamefont {D.}~\bibnamefont {Reitzner}},\
  }\href {\doibase 10.1103/PhysRevA.92.022115} {\bibfield  {journal} {\bibinfo
  {journal} {Phys. Rev. A}\ }\textbf {\bibinfo {volume} {92}},\ \bibinfo
  {pages} {022115} (\bibinfo {year} {2015})}\BibitemShut {NoStop}%
\bibitem [{\citenamefont {Guerini}\ \emph {et~al.}(2017)\citenamefont
  {Guerini}, \citenamefont {Bavaresco}, \citenamefont {Terra~Cunha},\ and\
  \citenamefont {Ac{\'i}n}}]{Leo2017}%
  \BibitemOpen
  \bibfield  {author} {\bibinfo {author} {\bibfnamefont {L.}~\bibnamefont
  {Guerini}}, \bibinfo {author} {\bibfnamefont {J.}~\bibnamefont {Bavaresco}},
  \bibinfo {author} {\bibfnamefont {M.}~\bibnamefont {Terra~Cunha}}, \ and\
  \bibinfo {author} {\bibfnamefont {A.}~\bibnamefont {Ac{\'i}n}},\ }\href
  {\doibase 10.1063/1.4994303} {\bibfield  {journal} {\bibinfo  {journal} {J.
  Math. Phys.}\ }\textbf {\bibinfo {volume} {58}},\ \bibinfo {pages} {092102}
  (\bibinfo {year} {2017})}\BibitemShut {NoStop}%
\bibitem [{\citenamefont {Oszmaniec}\ \emph {et~al.}(2017)\citenamefont
  {Oszmaniec}, \citenamefont {Guerini}, \citenamefont {Wittek},\ and\
  \citenamefont {Ac\'{\i}n}}]{OszmaniecPOVM}%
  \BibitemOpen
  \bibfield  {author} {\bibinfo {author} {\bibfnamefont {M.}~\bibnamefont
  {Oszmaniec}}, \bibinfo {author} {\bibfnamefont {L.}~\bibnamefont {Guerini}},
  \bibinfo {author} {\bibfnamefont {P.}~\bibnamefont {Wittek}}, \ and\ \bibinfo
  {author} {\bibfnamefont {A.}~\bibnamefont {Ac\'{\i}n}},\ }\href {\doibase
  10.1103/PhysRevLett.119.190501} {\bibfield  {journal} {\bibinfo  {journal}
  {Phys. Rev. Lett.}\ }\textbf {\bibinfo {volume} {119}},\ \bibinfo {pages}
  {190501} (\bibinfo {year} {2017})}\BibitemShut {NoStop}%
\bibitem [{\citenamefont {Oszmaniec}\ \emph {et~al.}(2018)\citenamefont
  {Oszmaniec}, \citenamefont {Maciejewski},\ and\ \citenamefont
  {Pucha{\l}a}}]{MOPOST}%
  \BibitemOpen
  \bibfield  {author} {\bibinfo {author} {\bibfnamefont {M.}~\bibnamefont
  {Oszmaniec}}, \bibinfo {author} {\bibfnamefont {F.~B.}\ \bibnamefont
  {Maciejewski}}, \ and\ \bibinfo {author} {\bibfnamefont {Z.}~\bibnamefont
  {Pucha{\l}a}},\ }\href {https://arxiv.org/abs/1807.08449} {\bibfield
  {journal} {\bibinfo  {journal} {arXiv:1807.08449}\ } (\bibinfo {year}
  {2018})}\BibitemShut {NoStop}%
\bibitem [{\citenamefont {Kleinmann}\ and\ \citenamefont
  {Cabello}(2016)}]{Kleinmann2016a}%
  \BibitemOpen
  \bibfield  {author} {\bibinfo {author} {\bibfnamefont {M.}~\bibnamefont
  {Kleinmann}}\ and\ \bibinfo {author} {\bibfnamefont {A.}~\bibnamefont
  {Cabello}},\ }\href {\doibase 10.1103/PhysRevLett.117.150401} {\bibfield
  {journal} {\bibinfo  {journal} {Phys. Rev. Lett.}\ }\textbf {\bibinfo
  {volume} {117}},\ \bibinfo {pages} {150401} (\bibinfo {year}
  {2016})}\BibitemShut {NoStop}%
\bibitem [{\citenamefont {D'Ariano}\ \emph {et~al.}(2005)\citenamefont
  {D'Ariano}, \citenamefont {Presti},\ and\ \citenamefont
  {Perinotti}}]{DAriano2005}%
  \BibitemOpen
  \bibfield  {author} {\bibinfo {author} {\bibfnamefont {G.~M.}\ \bibnamefont
  {D'Ariano}}, \bibinfo {author} {\bibfnamefont {P.~L.}\ \bibnamefont
  {Presti}}, \ and\ \bibinfo {author} {\bibfnamefont {P.}~\bibnamefont
  {Perinotti}},\ }\href {\doibase 10.1088/0305-4470/38/26/010} {\bibfield
  {journal} {\bibinfo  {journal} {J. Phys. A: Math. Gen.}\ }\textbf {\bibinfo
  {volume} {38}},\ \bibinfo {pages} {5979} (\bibinfo {year}
  {2005})}\BibitemShut {NoStop}%
\bibitem [{\citenamefont {Brandao}(2005)}]{Brandao_witness}%
  \BibitemOpen
  \bibfield  {author} {\bibinfo {author} {\bibfnamefont {F.~G.}\ \bibnamefont
  {Brandao}},\ }\href {\doibase 10.1103/PhysRevA.72.022310} {\bibfield
  {journal} {\bibinfo  {journal} {Phys. Rev. A}\ }\textbf {\bibinfo {volume}
  {72}},\ \bibinfo {pages} {022310} (\bibinfo {year} {2005})}\BibitemShut
  {NoStop}%
\bibitem [{\citenamefont {Hillery}(2016)}]{HilleryAlgo}%
  \BibitemOpen
  \bibfield  {author} {\bibinfo {author} {\bibfnamefont {M.}~\bibnamefont
  {Hillery}},\ }\href {\doibase 10.1103/PhysRevA.93.012111} {\bibfield
  {journal} {\bibinfo  {journal} {Phys. Rev. A}\ }\textbf {\bibinfo {volume}
  {93}},\ \bibinfo {pages} {012111} (\bibinfo {year} {2016})}\BibitemShut
  {NoStop}%
\bibitem [{\citenamefont {Bischof}\ \emph {et~al.}(2018)\citenamefont
  {Bischof}, \citenamefont {Kampermann},\ and\ \citenamefont
  {Bruss}}]{BrussPOVM}%
  \BibitemOpen
  \bibfield  {author} {\bibinfo {author} {\bibfnamefont {F.}~\bibnamefont
  {Bischof}}, \bibinfo {author} {\bibfnamefont {H.}~\bibnamefont {Kampermann}},
  \ and\ \bibinfo {author} {\bibfnamefont {D.}~\bibnamefont {Bruss}},\ }\href
  {https://arxiv.org/abs/1812.00018} {\bibfield  {journal} {\bibinfo  {journal}
  {arXiv:1812.00018}\ } (\bibinfo {year} {2018})}\BibitemShut {NoStop}%
\bibitem [{\citenamefont {Theurer}\ \emph {et~al.}(2018)\citenamefont
  {Theurer}, \citenamefont {Egloff}, \citenamefont {Zhang},\ and\ \citenamefont
  {B.}}]{PlenioEvol}%
  \BibitemOpen
  \bibfield  {author} {\bibinfo {author} {\bibfnamefont {T.}~\bibnamefont
  {Theurer}}, \bibinfo {author} {\bibfnamefont {D.}~\bibnamefont {Egloff}},
  \bibinfo {author} {\bibfnamefont {L.}~\bibnamefont {Zhang}}, \ and\ \bibinfo
  {author} {\bibfnamefont {M.~B.}\ \bibnamefont {Plenio.}},\ }\href {https://arxiv.org/abs/1806.07332}
  {\bibfield  {journal} {\bibinfo  {journal} {arXiv:1806.07332}\ } (\bibinfo
  {year} {2018})}\BibitemShut {NoStop}%
\bibitem [{\citenamefont {Bera}(2018)}]{BeraEvolution}%
  \BibitemOpen
  \bibfield  {author} {\bibinfo {author} {\bibfnamefont {M.}~\bibnamefont
  {Bera}},\ }\href {https://arxiv.org/abs/1809.02578} {\bibfield  {journal}
  {\bibinfo  {journal} {arXiv:1809.02578}\ } (\bibinfo {year}
  {2018})}\BibitemShut {NoStop}%
\bibitem [{\citenamefont {Vandenberghe}\ and\ \citenamefont
  {Boyd}(1996)}]{sdpREV}%
  \BibitemOpen
  \bibfield  {author} {\bibinfo {author} {\bibfnamefont {L.}~\bibnamefont
  {Vandenberghe}}\ and\ \bibinfo {author} {\bibfnamefont {S.}~\bibnamefont
  {Boyd}},\ }\href {https://epubs.siam.org/doi/abs/10.1137/1038003} {\bibfield  {journal} {\bibinfo  {journal} {SIAM
  Review}\ }\textbf {\bibinfo {volume} {38}},\ \bibinfo {pages} {49} (\bibinfo
  {year} {1996})}\BibitemShut {NoStop}%
\bibitem [{\citenamefont {Korzekwa}\ \emph {et~al.}(2018)\citenamefont
  {Korzekwa}, \citenamefont {Czachorski}, \citenamefont {Pucha{\l}a},\ and\
  \citenamefont {Zyczkowski}}]{KorzDist}%
  \BibitemOpen
  \bibfield  {author} {\bibinfo {author} {\bibfnamefont {K.}~\bibnamefont
  {Korzekwa}}, \bibinfo {author} {\bibfnamefont {S.}~\bibnamefont
  {Czachorski}}, \bibinfo {author} {\bibfnamefont {Z.}~\bibnamefont
  {Pucha{\l}a}}, \ and\ \bibinfo {author} {\bibfnamefont {K.}~\bibnamefont
  {Zyczkowski}},\ }\href {https://arxiv.org/abs/1812.09083} {\bibfield
  {journal} {\bibinfo  {journal} {arXiv:1812.09083}\ } (\bibinfo {year}
  {2018})}\BibitemShut {NoStop}%
\bibitem [{\citenamefont {Gurvits}\ and\ \citenamefont
  {Barnum}(2003)}]{Gurvits2004}%
  \BibitemOpen
  \bibfield  {author} {\bibinfo {author} {\bibfnamefont {L.}~\bibnamefont
  {Gurvits}}\ and\ \bibinfo {author} {\bibfnamefont {H.}~\bibnamefont
  {Barnum}},\ }\href {\doibase 10.1103/PhysRevA.68.042312} {\bibfield
  {journal} {\bibinfo  {journal} {Phys. Rev. A}\ }\textbf {\bibinfo {volume}
  {68}},\ \bibinfo {pages} {042312} (\bibinfo {year} {2003})}\BibitemShut
  {NoStop}%
\bibitem [{\citenamefont {Montanaro}(2007)}]{Montanaro2007}%
  \BibitemOpen
  \bibfield  {author} {\bibinfo {author} {\bibfnamefont {A.}~\bibnamefont
  {Montanaro}},\ }\href {\doibase 10.1007/s00220-007-0221-7} {\bibfield
  {journal} {\bibinfo  {journal} {Commun Math Phys}\ }\textbf {\bibinfo
  {volume} {273}},\ \bibinfo {pages} {619} (\bibinfo {year}
  {2007})}\BibitemShut {NoStop}%
\bibitem [{\citenamefont {Hayden}\ \emph {et~al.}(2004)\citenamefont {Hayden},
  \citenamefont {Leung}, \citenamefont {Shor},\ and\ \citenamefont
  {Winter}}]{Hayden2004}%
  \BibitemOpen
  \bibfield  {author} {\bibinfo {author} {\bibfnamefont {P.}~\bibnamefont
  {Hayden}}, \bibinfo {author} {\bibfnamefont {D.}~\bibnamefont {Leung}},
  \bibinfo {author} {\bibfnamefont {P.~W.}\ \bibnamefont {Shor}}, \ and\
  \bibinfo {author} {\bibfnamefont {A.}~\bibnamefont {Winter}},\ }\href
  {\doibase 10.1007/s00220-004-1087-6} {\bibfield  {journal} {\bibinfo
  {journal} {Communications in Mathematical Physics}\ }\textbf {\bibinfo
  {volume} {250}},\ \bibinfo {pages} {371} (\bibinfo {year}
  {2004})}\BibitemShut {NoStop}%
\bibitem [{\citenamefont {Matthews}\ \emph {et~al.}(2009)\citenamefont
  {Matthews}, \citenamefont {Wehner},\ and\ \citenamefont
  {Winter}}]{Matthews2009}%
  \BibitemOpen
  \bibfield  {author} {\bibinfo {author} {\bibfnamefont {W.}~\bibnamefont
  {Matthews}}, \bibinfo {author} {\bibfnamefont {S.}~\bibnamefont {Wehner}}, \
  and\ \bibinfo {author} {\bibfnamefont {A.}~\bibnamefont {Winter}},\ }\href
  {\doibase 10.1007/s00220-009-0890-5} {\bibfield  {journal} {\bibinfo
  {journal} {Communications in Mathematical Physics}\ }\textbf {\bibinfo
  {volume} {291}},\ \bibinfo {pages} {813} (\bibinfo {year}
  {2009})}\BibitemShut {NoStop}%
\bibitem [{\citenamefont {Tavakoli}\ \emph {et~al.}(2018)\citenamefont
  {Tavakoli}, \citenamefont {Smania}, \citenamefont {V\'ertesi}, \citenamefont
  {Rosset},\ and\ \citenamefont {Brunner}}]{ArminSelftest}%
  \BibitemOpen
  \bibfield  {author} {\bibinfo {author} {\bibfnamefont {A.}~\bibnamefont
  {Tavakoli}}, \bibinfo {author} {\bibfnamefont {M.}~\bibnamefont {Smania}},
  \bibinfo {author} {\bibfnamefont {T.}~\bibnamefont {V\'ertesi}}, \bibinfo
  {author} {\bibfnamefont {D.}~\bibnamefont {Rosset}}, \ and\ \bibinfo {author}
  {\bibfnamefont {N.}~\bibnamefont {Brunner}},\ }\href
  {https://arxiv.org/abs/1811.12712} {\bibfield  {journal} {\bibinfo  {journal}
  {arXiv:1811.12712}\ } (\bibinfo {year} {2018})}\BibitemShut {NoStop}%
\bibitem [{\citenamefont {Mironowicz}\ and\ \citenamefont
  {Paw{\l}owski}(2018)}]{PiotrPOVM}%
  \BibitemOpen
  \bibfield  {author} {\bibinfo {author} {\bibfnamefont {P.}~\bibnamefont
  {Mironowicz}}\ and\ \bibinfo {author} {\bibfnamefont {M.}~\bibnamefont
  {Paw{\l}owski}},\ }\href {https://arxiv.org/abs/1811.12872} {\bibfield
  {journal} {\bibinfo  {journal} {arXiv:1811.12872}\ } (\bibinfo {year}
  {2018})}\BibitemShut {NoStop}%
\bibitem [{\citenamefont {Uola}\ \emph {et~al.}(2018)\citenamefont {Uola},
  \citenamefont {Kraft}, \citenamefont {Shang}, \citenamefont {Yu},\ and\
  \citenamefont {Guhne}}]{GuhnePOVM}%
  \BibitemOpen
  \bibfield  {author} {\bibinfo {author} {\bibfnamefont {R.}~\bibnamefont
  {Uola}}, \bibinfo {author} {\bibfnamefont {T.}~\bibnamefont {Kraft}},
  \bibinfo {author} {\bibfnamefont {J.}~\bibnamefont {Shang}}, \bibinfo
  {author} {\bibfnamefont {X.-D.}\ \bibnamefont {Yu}}, \ and\ \bibinfo {author}
  {\bibfnamefont {O.}~\bibnamefont {Guhne}},\ }\href {\doibase 10.1103/PhysRevLett.122.130404}
  {\bibfield  {journal} {\bibinfo  {journal} {Phys. Rev. Lett.}\ }\textbf
  {\bibinfo {volume} {122}},\ \bibinfo {pages} {130404} (\bibinfo {year}
  {2019})}\BibitemShut {NoStop}%
\bibitem [{\citenamefont {R.}\ and\ \citenamefont {Regula}(2019)}]{Takagi2019}%
  \BibitemOpen
  \bibfield  {author} {\bibinfo {author} {\bibfnamefont {R.}~\bibnamefont
  {Takagi.}}\ and\ \bibinfo {author} {\bibfnamefont {B.}~\bibnamefont {Regula}},\
  }\href {https://arxiv.org/abs/1901.08127} {\bibfield  {journal} {\bibinfo  {journal} {arXiv:1901.08127}\
  } (\bibinfo {year} {2019})}\BibitemShut {NoStop}%
\bibitem [{\citenamefont {Carmeli}\ \emph {et~al.}(2018)\citenamefont
  {Carmeli}, \citenamefont {Heinosaari},\ and\ \citenamefont
  {Toigo}}]{TeikoComp}%
  \BibitemOpen
  \bibfield  {author} {\bibinfo {author} {\bibfnamefont {C.}~\bibnamefont
  {Carmeli}}, \bibinfo {author} {\bibfnamefont {T.}~\bibnamefont {Heinosaari}},
  \ and\ \bibinfo {author} {\bibfnamefont {A.}~\bibnamefont {Toigo}},\ }\href {\doibase 10.1103/PhysRevLett.122.130402}
  {\bibfield  {journal} {\bibinfo  {journal} {Phys. Rev. Lett.}\ }\textbf
  {\bibinfo {volume} {122}},\ \bibinfo {pages} {13040} (\bibinfo {year}
  {2019})}\BibitemShut {NoStop}%
\bibitem [{\citenamefont {Skrzypczyk}\ \emph {et~al.}(2019)\citenamefont
  {Skrzypczyk}, \citenamefont {\ifmmode \check{S}\else
  \v{S}\fi{}upi\ifmmode~\acute{c}\else \'{c}\fi{}},\ and\ \citenamefont
  {Cavalcanti}}]{CavalComp}%
  \BibitemOpen
  \bibfield  {author} {\bibinfo {author} {\bibfnamefont {P.}~\bibnamefont
  {Skrzypczyk}}, \bibinfo {author} {\bibfnamefont {I.}~\bibnamefont {\ifmmode
  \check{S}\else \v{S}\fi{}upi\ifmmode~\acute{c}\else \'{c}\fi{}}}, \ and\
  \bibinfo {author} {\bibfnamefont {D.}~\bibnamefont {Cavalcanti}},\ }\href {\doibase 10.1103/PhysRevLett.122.130403}
  {\bibfield  {journal} {\bibinfo  {journal} {Phys. Rev. Lett.}\ }\textbf
  {\bibinfo {volume} {122}},\ \bibinfo {pages} {130403} (\bibinfo {year}
  {2019})}\BibitemShut {NoStop}%
\bibitem [{\citenamefont {Cavalcanti}\ and\ \citenamefont
  {Skrzypczyk}(2016)}]{Skrz2016}%
  \BibitemOpen
  \bibfield  {author} {\bibinfo {author} {\bibfnamefont {D.}~\bibnamefont
  {Cavalcanti}}\ and\ \bibinfo {author} {\bibfnamefont {P.}~\bibnamefont
  {Skrzypczyk}},\ }\href {\doibase 10.1088/1361-6633/80/2/024001} {\bibfield
  {journal} {\bibinfo  {journal} {Reports on Progress in Physics}\ }\textbf
  {\bibinfo {volume} {80}},\ \bibinfo {pages} {024001} (\bibinfo {year}
  {2016})}\BibitemShut {NoStop}%
\bibitem [{\citenamefont {Gross}\ \emph {et~al.}(2009)\citenamefont {Gross},
  \citenamefont {Flammia},\ and\ \citenamefont {Eisert}}]{Gross2009}%
  \BibitemOpen
  \bibfield  {author} {\bibinfo {author} {\bibfnamefont {D.}~\bibnamefont
  {Gross}}, \bibinfo {author} {\bibfnamefont {S.~T.}\ \bibnamefont {Flammia}},
  \ and\ \bibinfo {author} {\bibfnamefont {J.}~\bibnamefont {Eisert}},\ }\href
  {\doibase 10.1103/PhysRevLett.102.190501} {\bibfield  {journal} {\bibinfo
  {journal} {Phys. Rev. Lett.}\ }\textbf {\bibinfo {volume} {102}},\ \bibinfo
  {pages} {190501} (\bibinfo {year} {2009})}\BibitemShut {NoStop}%
\end{thebibliography}

\section{Proofs of general results}\label{app:proofGEN}

For the convenience of the reader we first recall the definition of robustness for general convex and compact subset of free measurements $\F\subset\PP(d,n)$

\begin{equation} \label{eq:robDEF1} 
R_{\F}(\M):=\min \bigg\{s |\ \exists \N \text{ such that } \frac{\M+s\N}{1+s}\in \F \bigg\} \ .
\end{equation}

\setcounter{prop}{0} 
\setcounter{lem}{0}
\begin{prop}[Properties of robustness]\label{lem:prop1}
Let $\F\subset \PP(d,n)$ be a compact convex set of free measurements. The robustness  $R_\F$  satisfies the following properties: (i) 
Faithfulness: $R_\F (\M) \geq 0$ and $R_\F (\M)>0$ iff $\M\notin\F$, (ii) Convexity: $R_\F (p\M +(1-p) \M') \leq  p R_\F(\M) +(1-p) R_\F(\M')$ , (iii)  Monotonicity under free operations: $R_\F \left[\varphi(\M)\right]\leq R_\F (\M)$ for all free operations $\varphi\in\OO$.
\end{prop}
\begin{proof}
\ \\ 
\emph{(i) Faithfulness}. This property follows directly from the definition given in Eq.\eqref{eq:robDEF1}.
	
\noindent \emph{(ii) Convexity}.  From the definition of the robustness \eqref{eq:robDEF1} it follows that for two POVMs $\M,\M'$ there exist POVMs $\N,\N' \in \PP(d,n)$, such that  
	\begin{eqnarray}\label{pseudomix}
	 \M  &= (1+R_{\F}(\M))\tilde{\M}-R_{\F}(\M)\N\ , \\
	 \M'  &= (1+R_{\F}(\M'))\tilde{\M'}-R_{\F}(\M')\N'\ ,
	\end{eqnarray}
	where $\N,\N'\in\F$. We therefore get that for all $p\in[0,1]$
	\begin{equation}
	\begin{aligned}
	p\M +(1-p)\M' = p[(1+R_{\F}(\M))\tilde{\M}-R_{\F}(\M)\N] & \\ +(1-p)[(1+R_{\F}(\M'))\M'-R_{\F}(\M')\N'] \ .  
	\end{aligned}
	\end{equation}	
After defining 
	\begin{eqnarray}
       \tilde{r}&\coloneqq & pR_{\F}(\M)+(1-p)R_{\F}(\M')\ , \\
       \tilde{\N}&\coloneqq & \frac{p R_{\F}(\M)\N+(1-p)R_{\F}(\M')\N'}{\tilde{r}}\  , \\
       \textbf{T} &\coloneqq & \frac{p(\tilde{\M}+R_{\F}(\M)\tilde{\M})+(1-p)(\tilde{\M'}+R_{\F}(\M')\tilde{\M'})}{1+\tilde{r}}\ ,
	\end{eqnarray}
we note that $\tilde{N} \in \PP(d,n)$, while $\textbf{T}\in\F$. We finally observe that
	\begin{equation}
	   p\M +(1-p)\M' = (1+\tilde{r})\textbf{T}-\tilde{r}\tilde{\N}\ ,
	\end{equation}
	which directly implies $\tilde{r}=pR_{\F}(\M)+(1-p)R_{\F}(\M') \geq R_\F (p\M +(1-p)\M')$.
	
\noindent   \emph{(iii) Monotonicity under free operations}.  From the definition we have that for every $\M$ there exist $\N\in\PP(d,n)$ and $\tilde{\M}\in\F$ such that
\begin{equation}
\M  = (1+R_{\F}(\M))\tilde{\M}-R_{\F}(\M)\N\ .
\end{equation}
Applying $\varphi$  to both sides of the above equality we obtain
\begin{equation}\label{eq:mapACTS}
\varphi(\M) = (1+R_{\F}(\M))\varphi(\tilde{\M})-R_{\F}(\M)\varphi(\N)\ .
\end{equation}
Since, by assumption, $\varphi(\tilde{\M})\in\F$, equation \eqref{eq:mapACTS} directly implies $R_\F \left[\varphi(\M)\right]\leq R_\F (\M)$.
\end{proof}

\begin{lem}\label{lem:AltROB2}
Let $\M$ $\in$ $\PP(d,n)$ and $\F\subset\PP(d,n)$ be a compact convex set of free measurements. Then, the generalized robustness  $R_{\F}(\M)$ can be expressed as the following optimisation problem.
\begin{equation} \label{eq:AltRob2}
\begin{aligned}
& {\text{maximize}}
& & \sum_{i=1}^n \tr(Z_iM_i)-1 \\
& \text{subject to}
& &  Z_i \geq 0, \; i = 1, \ldots, n\ ,\\ 
& & & \sum_{i=1}^n Tr(Z_iN_i) \leq 1  \; \forall  \N  \in \F\ .
\end{aligned}
\end{equation}
\end{lem}
\begin{proof}
Our proof starategy is the following. The problem in \eqref{eq:AltRob2} belongs to the class of semi-infinite optimisation problems \cite{Brandao_witness}. We will compute its dual and argue that there is no duality gap. Finally, the dual problem will turn out to be equivalent to the original definition of robustness from \eqref{eq:robDEF1}. After simple algebraic manipulations we obtain that \eqref{eq:AltRob2} is equivalent to 
\begin{equation} \label{eq:secondPROB}
\begin{aligned}
-\min_\X \bigg\{\sum_i \tr(X_iM_i)-1: \forall {i}, \;  X_i \leq 0,\ \forall \N\in\F\ \ \sum_iTr(X_iN_i) \geq -1\bigg\}\ \\ 
=-\min_{\tilde{\X}} \bigg\{\sum_i \tr(\tilde{X}_iM_i):\forall{i},  \;  X_i \leq \frac{\mathbb{I}}{d}  \ \forall \N\in\F\ \   \sum_iTr(\tilde{X}_iN_i) \geq 0\bigg\}\ .
\end{aligned}
\end{equation}
The \textit{Lagrangian} corresponding to the last problem in \eqref{eq:secondPROB} reads 
\begin{equation}\label{eq:thirdproblem}
\begin{aligned}
\mathcal{L}(\tilde{X},\textbf{G},\lbrace k(\N) \rbrace)= \sum_{i=1}^{n}\bigg( \tr(\tilde{X}_iM_i)+\tr(X_i-\frac{I}{d})G_i-\int d\N k(\N) \tr(N_iX_i)\bigg) \ ,
\end{aligned}
\end{equation}
where $\textbf{G}=(G_1,\ldots,G_n)$ and $k(\N)$ are lagrange multipliers which satisfy $G_i \geq 0$ $\forall{i}$ and $ k\N)\geq 0$ for almost all $\N\in\F$. To construct the dual problem to \eqref{eq:thirdproblem} we minimize $\mathcal{L}(\tilde{X},\textbf{G},\lbrace k(\N) \rbrace)$ over primal variables $\tilde{\X}$. Since the \textit{Lagrangian} is affine function of $\tilde{\textbf{X}}$, we obtain minimum equal to $\infty$ unless $\textbf{G}+\M= \int k(\N)\N d\N$. If the condition holds then we obtain $-\frac{1}{d}\sum_{i=1}^{n}\tr (G_i)$. Hence we see that the dual of the problem \eqref{eq:AltRob2} can be written as 
\begin{equation}\label{n6}
\begin{aligned}
& {\text{minimize}}
& & \frac{1}{d} \sum_{i=1}^{n} \tr(G_i) \\
& \text{subject to}
& &  \textbf{G}+\M= \int k(\N)\N d\N \ , \\ 
& & & G_i \geq 0  \; \forall{i} \ , \\
& & & k(\N) \geq 0 \; \; \; \forall{\N}\in \F\ .
\end{aligned}
\end{equation}
We note that the primal problem \eqref{eq:AltRob2} posseses a strictly feasible points (for example $\mathbf{Z}=(\frac{\lambda \1}{d},\ldots,\frac{\lambda \1}{d}$, for $\lambda\in(0,1)$), which follows from the assumption of the compactness of $\F$ and the form of the constrains in \eqref{eq:AltRob2}. Finally, we conclude the proof by noting that $\textbf{G}+\M= \int k(\N)\N d\N $ iff $\textbf{G}=s \M'$ where $\M'\in \PP(d,n)$ and $s=\sum_i\frac{\tr G_i}{d}$.    
\end{proof}

\section{Proofs of results that concern particular classes of free measurements}\label{app:proofPART}

\subsection{Incoherent measurements}

 From the definition of robustness (see Eq.\eqref{eq:robDEF1}) and the definition of incoherent measurements we obtain that $R_{\IC}(\M)$ can be cast as the following SDP program

\begin{equation} \label{eq:SDPincoh}
\begin{aligned}
& {\text{minimize}}
& & s\\
& \text{s.t}
& & \frac{M_a+sN_a}{1+s}=\sum_i q(a|i)|i\rangle\langle i|  \; \forall a \ , \\ 
& & & N_a \geq 0, \; a = 1, \ldots, n. \; \;  \sum_a N_a=\mathbb{I}\ , \\
& & & \sum_a q(a|i)(1+s)=b \; \;  \forall i \ .
\end{aligned}
\end{equation}.

\setcounter{lem}{1}

\begin{lem}
[Dual characterization of robustness w.r.t. incoherent measurement] Robustness of a POVM $\M \in \PP(d,n)$ w.r.t set of incoherent measurement $\mathcal{IC}(d,n)$, \textit{i.e}, $R_{\IC}(\M)$ can be expressed as the solution of optimization problem given by 
\begin{equation} \label{DualincohLemApp}
\begin{aligned}
& {\text{maximize}}
& & \sum_{a=1}^{n}\tr(Z_aM_a) -1\\
& \text{subject to}
& & \forall i,a \; \; Z_a \geq 0\;\;\;\; \langle i |Z_a|i\rangle  = \langle i |Z_n|i\rangle \ ,  \\ 
& & & \tr(Z_n)=1\ .
\end{aligned}
\end{equation} 
\end{lem}
\begin{proof}
Introducing variables $\tilde{q}(a|i):= (1+s)q(a|i)$ $\forall i$ and $\forall a \in [1,(n-1)]$ we obtain an equivalent form of problem from \eqref{eq:SDPincoh} 
\begin{equation} \label{n16}
\begin{aligned}
& {\text{minimize}}
& & b-1\\
& \text{subject to}
& & -\sum_i \tilde{q}(a|i)|i\rangle\langle i| + M_a \leq 0 \; \; \;  \forall a \in [1,(n-1)]\ ,\\ 
& & & -b\mathbb{I}+\sum_i\sum_{a=1}^{n-1}\tilde{q}(a|i)|i\rangle\langle i| +M_n \leq 0\ .
\end{aligned}
\end{equation} 
To construct the dual of the optimization problem \cite{Skrz2016,sdpREV} given in \eqref{n16}, we start by writing down the \textit{Lagrangian} associated with primal problem 
\begin{equation}\label{Lagincoh}
\mathcal{L}(\textbf{Z},\tilde{q}(a|i)):=(b-1)+ \tr(Z_n(-b\mathbb{I}+\sum_i\sum_{a=1}^{n-1}\tilde{q}(a|i)|i\rangle\langle i| \\  + M_n))+ \sum_{a=1}^{n-1}\tr(Z_a(-\sum_i \tilde{q}(a|i)|i\rangle\langle i| + M_a ) \ ,
\end{equation}   
where we have introduced $\textbf{Z}=(Z_1,Z_2,\ldots,Z_n)$ as dual variable (Lagrange multipliers). To ensure that the \textit{Lagrangian} is always smaller than the objective function whenever the constraints of the \eqref{n16} are satisfied, we impose $Z_j \geq 0$ and $\langle i|Z_j|i \rangle=\langle i|Z_n|i \rangle$ $\forall{j \in [1,n]}$. To ensure that the \textit{Lagrangian} is independent of primal variables we have $\tr(Z_n)=1 $. Finally, the corresponding dual problem which sets the upper bound for primal optimization problem is given by 
\begin{equation} \label{eq:Dualincoh}
\begin{aligned}
& {\text{maximize}}
& & \sum_{a=1}^{n}\tr(Z_aM_a) -1\\
& \text{subject to}
& & \forall i, a \;\; Z_a \geq 0\;\;\;\; \langle i |Z_a|i\rangle  = \langle i |Z_n|i\rangle \ , \\ 
& & & \tr(Z_n)=1 \ .
\end{aligned}
\end{equation} 
By inspection one can conclude that strong duality holds, since   $\mathbf{Z}_0=(\frac{\mathbb{I}}{d},\ldots,\frac{\mathbb{I}}{d})$ satisfy the constraints in \eqref{eq:Dualincoh} and is positive-definite. Therefore, the solutions \eqref{n16} and \eqref{eq:Dualincoh} are equal.
\end{proof}

\subsection{Separable measurements}

\setcounter{prop}{3}

\begin{prop}[Maximal advantage of entangled measurements over separable measurements] \label{lem:relPOWsep}
Let $d_A$ and $d_B$ denote the dimensions of local spaces in the biparticle Hilbert space $\H_{AB}$. Then, the following inequalities hold
\begin{equation}\label{eq:UBbipart2}
\min\lbrace d_A,d_B\rbrace -1 \leq \max_{\M} R_{\Sep(AB)}(\M)  \leq \min\lbrace d_A,d_B\rbrace\ , 
\end{equation} 
where the  maximization is over all measurements on $\H_{AB}$. \\
Moreover, the maximal robustness of entangled qubit measurement increases exponentially with the size of the system. Specifically, for sufficiently large $N$ we have 
\begin{equation}\label{eq:UBmultipart2}
c \frac{2^N}{8N^2} \leq \max_{\M} R_{\Sep(N)}(\M)  \leq 2^{\frac{3}{2}N -1} -1\ , 
\end{equation} 
where the  maximization is over all measurements on $\H_N$ and $c> 0.7$.
\end{prop}
\begin{proof}[Proof of lower bounds in Eq.\eqref{eq:UBbipart2} and Eq.\eqref{eq:UBmultipart2}]
Our strategy to prove the lower bounds is to construct ensembles $\E$ for which the relative advantage 
\begin{equation}
\frac{\psucc(\E,\M)}{{\max_{\N\in Sep(\mathcal{H}_{A}:\mathcal{H}_{B})}}\psucc(\E,\N)}
\end{equation}
of some entangled measurement $\M$ over the separable ones will be large. 
\emph{(i) Biparticle case} Let $D=min\lbrace d_A,d_B\rbrace$. Following Nathanson \cite{Nathanson2005}, we consider an uniform  ensemble $\E_0=\lbrace{\frac{1}{D^2},\Psi_{nm}\rbrace}_{n,m=1}^d $ consisting orthogonal maximally entangled states from $\H_A \otimes \H_B$,
\begin{equation}\label{eq:states}
    \ket{\Psi_{nm}}=\frac{1}{\sqrt{d}}\sum_{j=0}^{d-1}e^{\frac{i2\pi jn}{d}}|j\rangle \otimes |(j+m) \; \; \text{mod } d\rangle\ .
\end{equation}
Since vectors $\ket{\Psi_{nm}}$ are orthogonal, they can be distinguished perfectly by measurement $\M_0$ whose effects are of the form $M_{nm}=\Psi_{nm}+R_{nm}$, where $R_{nm}$ are operators orthogonal to states $\ket{\Psi_{nm}}$. On the other hand there is a limit of distinguishability  of states from $E_0$ via separable measurements because of the well-known property (see for example Lemma 2 of \cite{BandSep}) stating that for all biparticle quantum states $\ket{Psi}\in\H_A \otimes \H_B$ and positive separable operators $T$
\begin{equation}\label{eq:usefunINEQ}
\tr( \Psi T )\leq \lambda^2 \tr(T)\ ,
\end{equation}
where $\lambda$ is the greatest Schmidt coefficient of $\ket{\Psi}$. Let $\mathbb{P}_0$ be a projector onto a subspace $\tilde{H}_{AB}\subset\H_{AB}$ spanned by vectors $\ket{i}\ket{j}$, where $i,j=0,\ldots,D-1$. Importantly, $\mathbb{P}_0$ is a separable operator. Applying inequality \eqref{eq:usefunINEQ} to our problem we obtain that for all $D$-outcome separable measurements $\N$ we can upper bound $\psucc$ as follows
\begin{equation}
\psucc(\E_0,\N)=\frac{1}{D^2}\sum_{n,m=1}^D \tr(\Psi_{nm} N_{nm}) =\frac{1}{D^2} \sum_{n,m=1}^D \tr(\Psi_{nm} \mathbb{P}_0 N_{nm} \mathbb{P}_0)\leq \frac{1}{D}\frac{1}{D^2} \sum_{n,m=1}^D \tr(\mathbb{P}_0 N_{nm} \mathbb{P}_0)=\frac{1}{D}\ ,
\end{equation}
where we have used the separability of operators $\mathbb{P}_0 N_{nm} \mathbb{P}_0$ and the fact that states $\Psi_{nm}$ are maximally entangled and hence $\lambda=1/\sqrt{D}$. Using the operational interpretation of the robustness $R_{\Sep(AB}$ we conclude the proof by observing that
\begin{equation}\label{eq:randomENSAMBL}
R_{\Sep(AB)}(\M_0)\geq \frac{\psucc(\E_0,\M_0)}{\max{N\in\Sep(AB)} \psucc(\E_0,\N)}-1\geq D -1\ .
\end{equation}
\\
\emph{(ii)} {\emph{Multiparticle case}} We proceed analogously to the biparticle case. Consider a random  uniform ensemble of $M=2^N$ Haar random, independant and  identically distributed pure states in $\H_N$, 
\begin{equation}
\E_{\mathrm{Haar}}=\left\lbrace{\frac{1}{2^N},\Psi_i}\right\rbrace_{i=1}^{2^N} \ .
\end{equation} 
The above ensemble is random as states composing it are independent random pure states i.e. $\Psi_i = U_i \Psi_0 U^\dagger_i$, for $U_i$  distributed according to Haar measure on the group $\mathrm{U}(\H_N)$. In \cite{Montanaro2007} it was shown that in the limit of large $N$, for typical realizations of the ensemble $\E_{\mathrm{Haar}}$ the so-called pretty good measurement $\M_{PGM}(\E_{\mathrm{Haar}})$   is capable to to distinguish states from $\E_{\mathrm{Haar}}$ with success probability bounded away form zero
\begin{equation}\label{eq:PGM}
\psucc(\E_{\mathrm{Haar}},\M_{PGM}(\E_{\mathrm{Haar}}))\geq 0.7\ .
\end{equation}
Importantly, since states $\Psi_i$ are Haar random, they are typically highly entangled and this causes them to be indistinguishable by separable measurements. To show this we will use a result that can be found in the proof of Theorem 2  form \cite{Gross2009}, which states that for Haar random pure state $\Psi$ with probability $P\leq\exp(-N^2)$
\begin{equation}
\max_{\phi\in\mathrm{SEP}} \tr(\phi \Psi) \geq \frac{8N^2}{2^N}\ , 
\end{equation}
where the maximum is over all separable pure states on $\H_N$. Now, since states $\Psi_i$ are iid Haar-random, by the union bound, we get that with probability $q\geq 1 - 2^N \exp(-N^2)$, in the limit of large dimensions
\begin{equation}\label{eq:boundOVERLAP}
\forall_i\ \max_{\phi\in\mathrm{SEP}} \tr(\phi \Psi_i) \leq \frac{8N^2}{2^N}\ .
\end{equation}
The above bound (which holds with high probability as $N\rightarrow\infty$) can be the used to upper bound the success probability of discrimination of states from the the ensemble. Note that for every separable effect $N_i$ we have $\tr(\Psi_i N_i)\leq \tr(N_i) \max_{\phi\mathrm{SEP}} \tr(\phi \Psi_i)\leq \tr(N_i) \frac{8N^2}{2^N} $. Using this we obtain
\begin{equation}\label{eq:finalUPPER}
\psucc(\E_{\mathrm{Haar}},\N)\leq\frac{1}{2^N}\sum_{i=1}^{2^N}\tr(N_i) \frac{2^N}{8N^2}=\frac{8N^2}{2^N}\ .
\end{equation}  
Finally, combining \eqref{eq:PGM}, \eqref{eq:finalUPPER}, and the operational characterisation of $R_{\Sep(N)}$ we obtain the desired lower bound by realizing that there exist a realization of a random ensemble $\E_{\mathrm{Haar}}$ such that
\begin{equation}
\max_\M R_{\Sep(N)}(\M)\geq \frac{\psucc(\E_{\mathrm{Haar}},\M_{PGM}(\E_{\mathrm{Haar}}))}{\max_{\N\in\Sep(N)}\psucc(\E_{\mathrm{Haar}},\N)}-1\geq c \frac{2^N}{8N^2}-1\ .
\end{equation}

\end{proof}

 
\end{document}